\documentclass[aps,prl,superscriptaddress,longbibliography,showpacs,twocolumn]{revtex4-2}

\usepackage{pdfpages}

\makeatletter
\AtBeginDocument{\let\LS@rot\@undefined}
\makeatother

\usepackage{soul}
\setstcolor{red}
\usepackage{comment}

\usepackage{lipsum}
\usepackage{caption}
\DeclareCaptionJustification{justified}{\leftskip=0pt \rightskip=0pt \parfillskip=0pt plus 1fil}

\usepackage{graphicx}
\usepackage{dcolumn}
\usepackage{bm}
\usepackage{caption}
\usepackage{amsmath,amssymb,amsthm}
\usepackage{braket}
\usepackage{empheq}
\usepackage{subfig}
\usepackage{tikz}
\usepackage{url}
\usepackage{hyperref}
\usepackage{quantikz}
\usepackage{xcolor}
\usepackage{float}
\usepackage{mathrsfs}
\DeclareMathOperator{\Tr}{Tr}

\definecolor{blueviolet}{rgb}{0.2, 0.2, 0.6}
\definecolor{webgreen}{rgb}{0,.5,0}
\definecolor{webbrown}{rgb}{.6,0,0}

\hypersetup{
    breaklinks=true,
    colorlinks=true,       
    linkcolor=blueviolet,          
    citecolor=webgreen,        
    filecolor=magenta,      
    urlcolor=webbrown,           
    runcolor=cyan
}

\newtheorem{theorem}{Theorem}
\newtheorem*{theorem*}{Theorem}
\newtheorem{corollary}{Corollary}
\newtheorem{lemma}{Lemma}
\newtheorem{proposition}{Proposition}

\newtheorem{problem}{Problem}

\usepackage[ruled,vlined]{algorithm2e}
\SetKwFunction{FProp}{Propagate}
\SetKwFunction{FOptP}{OptMix}
\SetKwFunction{FGR}{GRAPE}

\usepackage{etoolbox}
\makeatletter
\patchcmd{\@citex}{\unskip, et al.}{\unskip et al.}{}{}
\makeatother

\begin{document}

\title{Randomized Quantum Optimal Control}
\author{Leeseok Kim}
\affiliation{Center for Quantum Information and Control, University of New Mexico, NM 87131, USA}
\affiliation{Department of Electrical and Computer Engineering, University of New Mexico, NM 87131, USA}
\author{Francisco Riberi}
\affiliation{Center for Quantum Information and Control, University of New Mexico, NM 87131, USA}
\affiliation{Department of Electrical and Computer Engineering, University of New Mexico, NM 87131, USA}
\author{Kevin Baca}
\affiliation{Center for Quantum Information and Control, University of New Mexico, NM 87131, USA}
\affiliation{Department of Physics and Astronomy, University of New Mexico, NM 87131, USA}
\author{Milad Marvian}
\affiliation{Center for Quantum Information and Control, University of New Mexico, NM 87131, USA}
\affiliation{Department of Electrical and Computer Engineering, University of New Mexico, NM 87131, USA}
\affiliation{Department of Physics and Astronomy, University of New Mexico, NM 87131, USA}

\begin{abstract}
Quantum optimal control (QOC) aims to find control functions that optimally steer a quantum system toward a target operation. We introduce a \emph{randomized} QOC framework {where optimization is carried over an ensemble of control functions and their probabilities, instead of a single set of functions}. Using this framework, we prove that randomized QOC can reach a target accuracy faster than any deterministic protocol under the same resource constraints. We also develop general symmetry-based constructions that convert a given control into an ensemble of controls that can systematically reduce the error. We benchmark these constructions for CNOT implementation and find that the resulting randomized protocol quadratically suppresses the error of the optimized deterministic solution. In addition, we introduce randomized GRAPE, which generalizes GRAPE to directly optimize control ensembles and their associated probabilities. Finally, as a related application, we discuss randomized boundary-pulse constructions {that enhance} robustness against coherent noise.
\end{abstract}

\maketitle

\textit{Introduction---} Precise control of quantum systems {has proven} essential for advancing quantum technologies. Given a system model, quantum optimal control (QOC) \cite{werschnikquantum2007,glaser2015training,d2021introduction,koch2022quantum,ansel2024introduction} provides analytical \cite{khaneja2001time,d2001optimal,boscain2002optimal,ggaron2013time,hegerfeldt2013driving,boscain2021introduction} and numerical \cite{khaneja2005optimal,palao2003optimal,caneva2011chopped,rach2015dressing,machnes2018tunable,muller2022one} techniques for designing control functions that realize a desired operation. QOC has become a key tool for implementing high-fidelity gates across quantum computing platforms \cite{kelly2014optimal,dodle2014high,heeres2017implementing,figgatt2019parallel,omran2019generation,jandura2022time,evered2023high}.

In parallel, randomization has become a powerful tool for improving quantum protocols, with applications in quantum algorithms \cite{campbell2019random,childs2019faster,boixo2009eigenpath,ouyang2020compilation,wan2022randomized,martyn2025halving}, noise tailoring \cite{wallman2016noise,hashim2021randomized}, and dynamical decoupling \cite{viola2005random,kern2009randomized,yi2026faster}. By sampling from an ensemble of protocols rather than repeatedly applying a fixed deterministic one, randomized schemes can average out coherent errors and, in some cases, outperform the best deterministic protocol within a given class. Yet most such methods are formulated in the gate-based model, leaving the systematic integration of randomization into continuous-time models largely unexplored.

In this Letter, we propose a general framework for randomized QOC, in which one optimizes over an ensemble of control waveforms and their associated probabilities, rather than over a single waveform (cf.~Fig.~\ref{fig:main-figure}). Randomized QOC expands the feasible set from deterministic trajectories to convex combinations of trajectories. This probabilistic mixing enables coherent-error cancellations between different control branches, allowing randomized protocols to outperform the deterministic optimum under the same resource constraints.

{We first prove, through an exactly solvable single-qubit example, that randomized QOC can outperform the deterministic optimum under identical control constraints. The optimal randomized error is the square of the optimal deterministic error, implying that a target accuracy can be reached in less time than with any deterministic protocol.}
{Motivated by this mechanism, we develop two symmetry-based prescriptions that transform a deterministic control into a randomized ensemble with suppressed leading-order error.}

\begin{figure}[t!]
    \centering
    \includegraphics[width=8.6cm]{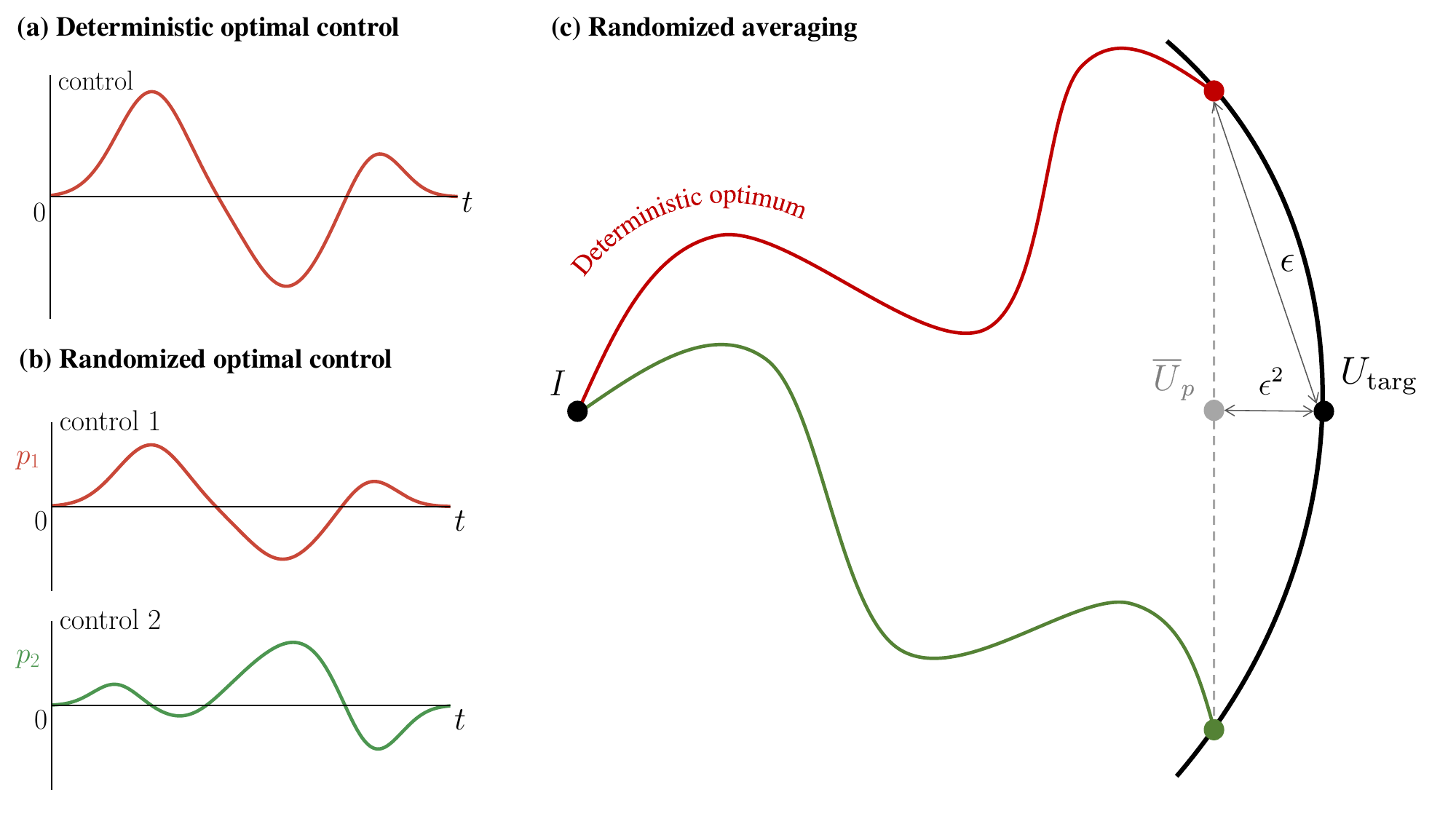}
    \caption{Schematic comparison between deterministic and randomized quantum optimal control. (a) Conventional QOC optimizes a single control waveform to approximate the target unitary $U_{\rm targ}$. (b) Randomized QOC optimizes an ensemble of control waveforms, shown here with two branches, together with the associated probability distribution. (c) Averaging over the unitaries generated by different controls can yield a quadratically closer approximation to $U_{\rm targ}$ than the deterministic optimum.}
    \label{fig:main-figure}
\end{figure}

{We benchmark these prescriptions on CNOT implementation in a two-qubit model. The resulting randomized ensembles exhibit the same quadratic suppression achieved by the single-qubit example and closely match the performance of a randomized generalization of the GRAPE  algorithm \cite{khaneja2005optimal}. }
Finally, building on Ref.~\cite{yi2026faster}, we extend the framework to noise-robust control through randomized boundary-pulse schemes. In particular, we {provide} a general randomized construction {capable of canceling arbitrary} first-order 1-local coherent noise {for Clifford targets.}

\textit{Randomized-control framework---}
We consider a controlled quantum system
\begin{align}
\label{model}
H_{\bm f}(t)=H_0+\sum_{j=1}^m f_j(t)H_j ,
\end{align}
where $H_0$ is the drift Hamiltonian, $H_j$ are the control Hamiltonians, and $f_j(t)$ are control functions, {possibly subject to boundedness constraints}. We collect these functions into the vector-valued control waveform $\bm f(t)=(f_1(t),\ldots,f_m(t))$, and {write} the corresponding unitary generated over a fixed time $T$,
\begin{align}
U_{\bm f}(T)
=
\mathcal T
\exp\left(
-i\int_0^T H_{\bm f}(t)dt
\right),
\end{align}
where $\mathcal T$ denotes time ordering. Conventional QOC seeks a single control waveform $\bm f$ whose final unitary $U_{\bm f}(T)$ best approximates a target unitary $U_{\rm targ}$.

In the randomized-control framework, we instead optimize over an ensemble of $M$ waveforms
$\{\bm f^{(i)}\}_{i=1}^M$ and a probability distribution $\bm p=(p_1,\ldots,p_M)$, with $p_i\ge0$ and $\sum_{i=1}^M p_i=1$. Writing $U^{(i)}(T):=U_{\bm f^{(i)}}(T)$, the ensemble induces the {following} mixed-unitary channel:
\begin{align}
\mathcal R_{\bm p}(\rho) := \sum_{i=1}^M p_iU^{(i)}(T)\rho U^{(i)}(T)^\dagger .
\end{align}
{Given a target channel $\mathcal U_{\rm targ}(\rho):= U_{\rm targ}\rho U_{\rm targ}^\dagger$, randomized QOC seeks to  minimize:}
\begin{align}
\label{diamond-distance-rqoc}
\min_{\{\bm f^{(i)}\}_{i=1}^M,\bm p} \frac12 \left\| \mathcal R_{\bm p} - \mathcal U_{\rm targ} \right\|_\diamond,
\end{align}
where $\|\mathcal A-\mathcal B\|_\diamond$ denotes the diamond distance between quantum channels $\mathcal A$ and $\mathcal B$ \cite{watrous2018the}.

\textit{Mechanism of randomized advantage---} {The advantage of randomized control originates from coherent-error cancellation between different branches $\boldsymbol{f}^{(i)}$. This is conveniently described in an averaged-unitary picture.} Since each branch represents the same channel up to a global phase, we may choose branch phases $\bm\phi=(\phi_1,\ldots,\phi_M)$ and define the average
\begin{align}
\overline U_{\bm p,\bm\phi} := \sum_{i=1}^M p_i e^{i\phi_i}U^{(i)}(T).
\end{align}
By the sharpened mixing lemma \cite{chen2021concentration}, the phase-minimized averaged-unitary error upper bounds the diamond-distance error:
\begin{align}
\epsilon_\diamond(\bm p,\{\bm f^{(i)}\}) = \frac12 \left\| \mathcal R_{\bm p} - \mathcal U_{\rm targ} \right\|_\diamond \le 2 \min_{\bm \phi} \left\| \overline U_{\bm p,\bm \phi} - U_{\rm targ}
\right\|.
\end{align}
where $\|A\|$
denotes the operator norm.

Let $\bm f$ implement
\begin{align}
\label{reference_error}
U_{\bm f}(T) = U_{\rm targ}e^{-iE}, \qquad \|E\|\ll1.
\end{align}
{Here $\bm f$ may denote any admissible waveform, including a deterministic optimum.} For a randomized protocol, choosing phase representatives such that
\begin{align}
e^{i\phi_i}U^{(i)}(T) = U_{\rm targ}e^{-iE_i},
\qquad \|E_i\|\ll1,
\end{align}
the phase-aligned averaged unitary satisfies
\begin{align}
\label{averaged-unitary-phase}
\overline U_{\bm p,\bm\phi} = U_{\rm targ} \left(I - i\sum_{i=1}^M p_iE_i + \mathcal O(\|E_i\|^2) \right).
\end{align}
Thus, averaging replaces the error generator of a single implementation by the ensemble-averaged error generator $\overline E:=\sum_{i=1}^M p_iE_i$. {The possibility of engineering cancellations in $\overline E$ is the origin of the randomized advantage. By choosing the ensemble $(\{\bm f^{(i)}\}_{i=1}^M,\bm p)$ so that the dominant components of the $E_i$ cancel in $\overline E$, the leading coherent error can be eliminated. In particular, $\overline E=0$ implies quadratic error suppression.}

\textit{Randomized advantage in approximate time-optimal control---}
Randomization can reduce the time required to reach a prescribed target accuracy. At fixed tolerance $\epsilon>0$, define $T_{\epsilon_\diamond}^{\rm det}$ and $T_{\epsilon_\diamond}^{\rm rand}$ as the minimum times needed to reach the target channel $\mathcal U_{\rm targ}$ within diamond-distance $\epsilon$ using deterministic and randomized controls, respectively. {We then have}:
\begin{theorem}
\label{thm:strict-randomized-speedup}
There exists a Hamiltonian model of the form \eqref{model}, a target unitary channel $\mathcal U_{\rm targ}$, and a tolerance $\epsilon>0$ such that, under the same control constraints,
\begin{align}
T_{\epsilon_\diamond}^{\rm rand} < T_{\epsilon_\diamond}^{\rm det}.
\end{align}
\end{theorem}

\textit{Analytically solvable example---}
We present a minimal exactly solvable example in which the optimal randomized error is the square of the optimal deterministic error, {with technical details given in Supplemental Material (SM).} This provides a constructive proof of Theorem~\ref{thm:strict-randomized-speedup}.

Consider the single-qubit model
\begin{align}
\label{sq-model}
H(t)=f(t)X+Z, \qquad |f(t)|\le f_{\max},
\end{align}
{with $\{X,Y,Z\}$ the relevant Pauli matrices and target} $U_{\rm targ}=e^{-i\pi/2Y}$. We work in the short-time regime $T\sqrt{1+f_{\max}^2}<\pi$, {where} the target is not exactly reachable, and assume $f_{\max}\le1$. In this regime, the deterministic optimum is
\begin{align}
\label{f_plus_protocol}
f_+(t)=
\begin{cases}
+f_{\max}, & 0\le t<T/2, \\
-f_{\max}, & T/2\le t\le T.
\end{cases}
\end{align}

Now consider the sign-flipped control $f_-(t):=-f_+(t)$. An optimal randomized protocol solving {Eq.} \eqref{diamond-distance-rqoc} is 
\begin{align}
\label{identity_randomized_protocol}
f^{(1)}(t)=f_+(t), \quad f^{(2)}(t)=f_-(t), \quad p_1=p_2=\frac12.
\end{align}

Let $\epsilon_{\diamond}^{\rm det,\star}(T)$ and $\epsilon_{\diamond}^{\rm rand,\star}(T)$ denote the optimal diamond-distance errors achieved by the deterministic protocol \eqref{f_plus_protocol} and the randomized protocol \eqref{identity_randomized_protocol}, respectively. Direct evaluation yields the exact relation
\begin{align}
\epsilon_{\diamond}^{\rm rand,\star}(T) = \left(\epsilon_{\diamond}^{\rm det,\star}(T)
\right)^2 < \epsilon_{\diamond}^{\rm det,\star}(T),
\end{align}
demonstrating a quadratic separation between randomized and deterministic optimal control. 

{This example illustrates the central mechanism enabling a randomization advantage: the deterministic optimum is not improved by finding a better control waveform, but by enlarging the feasible set from a single waveform to a convex combination of optimal waveforms whose coherent errors cancel.}

The speedup follows immediately. Choose any time $T_0$ in the certified short-time regime and any tolerance $\epsilon$ satisfying $\epsilon_{\diamond}^{\rm rand,\star}(T_0) < \epsilon < \epsilon_{\diamond}^{\rm det,\star}(T_0)$. Since $\epsilon_{\diamond}^{\rm det,\star}(T)$ decreases monotonically in this regime, no deterministic protocol can reach tolerance $\epsilon$ for any $T\le T_0$, whereas the randomized protocol does so at $T_0$. Hence $T_{\epsilon_\diamond}^{\rm rand} \le T_0 < T_{\epsilon_\diamond}^{\rm det}$, which proves Theorem~\ref{thm:strict-randomized-speedup}.

The origin of the quadratic separation is transparent in the phase-aligned averaged-unitary picture. The sign flip $f_+\mapsto f_-=-f_+$ {can be regarded as a consequence of} conjugating the Hamiltonian with $Z$:
\begin{align}
U_{f_-}(T) = ZU_{f_+}(T)Z.
\end{align}
If $U_{f_+}(T)=U_{\rm targ}e^{-iE}$, then $ZEZ=-E$ in this example. Since $ZU_{\rm targ}Z=-U_{\rm targ}$, we have $-U_{f_-}(T) = U_{\rm targ}e^{+iE}$. Thus, the two branches realize the same target unitary up to a global phase, while carrying opposite coherent errors, $E_1=E$ and $E_2=-E$, so equal mixing yields $\overline E=0$. By Eq.~\eqref{averaged-unitary-phase}, the first-order error vanishes and only quadratic corrections remain.

\textit{Symmetry-generated randomized control---}
The single-qubit example shows that the randomized advantage can arise from averaging control branches with opposite coherent errors. This suggests a constructive route {to improving performance:} rather than optimizing all branches independently, we start from a reference control $\bm f$, for instance a deterministic QOC solution, and generate additional branches using symmetries of the Hamiltonian model. We present two general constructions {that implement this procedure}: a target-preserving twirl and a time-reversal pairing.

\textit{(i) Target-preserving twirl.}
The first construction cancels error components that transform nontrivially under a symmetry {that preserves} the target operation. Let $\mathsf G$ be a finite set of unitaries {commuting} with the target unitary up to a global phase,
\begin{align}
\label{target_preserving_G}
g^\dagger U_{\rm targ}g=e^{i\theta_g}U_{\rm targ}, \qquad g\in\mathsf G .
\end{align}
Assume that each conjugation by $g\in\mathsf G$ is implementable through an admissible transformed control $\bm f^{(g)}$, satisfying
\begin{align}
H_{\bm f^{(g)}}(t)=g^\dagger H_{\bm f}(t)g .
\end{align}
The resulting evolution obeys
\begin{align}
U_{\bm f^{(g)}}(T)=g^\dagger U_{\bm f}(T)g ,
\end{align}
and the randomized protocol samples uniformly over the resulting symmetry-related branches.

Let $E$ denote the generator of the relative unitary error for the reference implementation, defined in Eq. \eqref{reference_error}. Using Eq.~\eqref{target_preserving_G}, each transformed branch obeys
\begin{align}
U_{\bm f^{(g)}}(T) = e^{i\theta_g}U_{\rm targ}e^{-ig^\dagger E g}.
\end{align}
Aligning branch phases and averaging uniformly over $g\in\mathsf G$, Eq.~\eqref{averaged-unitary-phase} gives
\begin{align}
\overline U_{\bm p,\bm\phi} = U_{\rm targ} \left(I-i\overline E_{\mathsf G} + \mathcal O(\|E\|^2) \right),
\end{align}
{with $\overline E_{\mathsf G} := \frac{1}{|\mathsf G|} \sum_{g\in\mathsf G}g^\dagger E g $ the twirled first-order error.}
{It follows that} all first-order error components that are modified by the twirl are canceled, while {twirl-invariant contributions} remain.

\emph{(ii) Time-reversal pairing.}
To suppress residual errors that are invariant under the twirl, we introduce a complementary construction that pairs the reference deterministic implementation with a transformed inverse evolution. This transformation reverses the sign of the relevant residual error contributions, so averaging the two implementations cancels them to first order.

Consider a Hermitian target unitary,
\begin{align}
\label{hermitian_target}
U_{\rm targ}^\dagger=U_{\rm targ},
\end{align}
and choose a finite set $\mathsf V$ of unitaries that commute with the target unitary up to a global phase,
\begin{align}
\label{target_preserving_V}
v^\dagger U_{\rm targ}v=e^{i\theta_v}U_{\rm targ},
\qquad v\in\mathsf V .
\end{align}
For the given reference control $\bm f$, {we restrict to the class of $v\in\mathsf V$ admitting} a transformed control $\widetilde{\bm f}^{(v)}$ that satisfies
\begin{align}
\label{implementable_time_reversal}
H_{\widetilde{\bm f}^{(v)}}(t) = -v^\dagger H_{\bm f}(T-t)v.
\end{align}
Time reversal generates branches implementing the inverse evolution with conjugated error generators. Specifically,
\begin{align}
\label{endpoint_time_reversal}
U_{\widetilde{\bm f}^{(v)}}(T) = v^\dagger U_{\bm f}(T)^\dagger v = e^{i\theta_v} U_{\rm targ} e^{i\Phi_v(E)},
\end{align}
where $\Phi_v(E) := U_{\rm targ}^\dagger v^\dagger E v U_{\rm targ}$.

We average the reference branch with the transformed inverse branches using equal total weights. Aligning branch phases and applying Eq.~\eqref{averaged-unitary-phase} gives
\begin{align}
\label{time_reversal_pair_average}
\overline U_{\bm p,\bm \phi} &= U_{\rm targ}\left(I-\frac{i}{2} \left(E-\overline\Phi_{\mathsf V}(E)\right) + \mathcal O(\|E\|^2) \right),
\end{align}
where $\overline\Phi_{\mathsf V}(E) := \frac{1}{|\mathsf V|} \sum_{v\in\mathsf V} \Phi_v(E)$. {From Eq.~\eqref{time_reversal_pair_average}, the first-order error cancels whenever}
\begin{align}
{\overline\Phi_{\mathsf V}(E)=E.}
\end{align}
A simple sufficient condition is
\begin{align}
\label{time_reversal_commuting_condition}
[E,U_{\rm targ}]=0, \qquad [E,v]=0 \quad (v\in\mathsf V).
\end{align}

The two constructions are complementary: the twirl removes symmetry-odd errors, while time-reversal pairing cancels residual twirl-invariant contributions. We now illustrate them in a two-qubit model.

\textit{(iii) Example.}
We consider the two-qubit {Hamiltonian}
\begin{align}
\label{two_qubit_model}
H_{\bm f}(t) = JZ_1Z_2 +  \sum_{j=1}^2 \left(f_{x,j}(t)X_j + f_{y,j}(t)Y_j \right),
\end{align}
where
$\bm f(t)=(f_{x,1}(t),f_{y,1}(t),f_{x,2}(t),f_{y,2}(t))$ is a fixed reference waveform, for instance, a deterministic QOC solution. We take the target to be $U_{\rm targ}={\rm CNOT}_{1\to2}$.

For the target-preserving twirl, choose $\mathsf G=\{I,Z_1\}$, which satisfies Eq.~\eqref{target_preserving_G}. In the model of Eq.~\eqref{two_qubit_model}, conjugation by $Z_1$ leaves the drift $Z_1Z_2$ invariant and flips the sign of the local controls on the first qubit. Thus the transformed controls are
\begin{align}
\bm f^{(I)}(t) &= \bm f(t), \\
\bm f^{(Z_1)}(t) &= (-f_{x,1}(t),-f_{y,1}(t),f_{x,2}(t),f_{y,2}(t)).
\end{align}
{The twirled unitary branch $U_{\bm f^{(Z_1)}}(T)=Z_1U_{\bm f}(T)Z_1$ follows}, and the averaged error generator becomes $\overline E_{\mathsf G} = \frac12\left( E+Z_1EZ_1 \right)$. Hence the twirl cancels, to first order, all Pauli error terms containing $X_1$ or $Y_1$, {leaving errors such as $Z_1X_2$, $Z_1$, and $X_2$ invariant}.

{To cancel} such residual error components, we use the time-reversal pairing with
$\mathsf V=\{X_2\}$, which satisfies Eq.~\eqref{target_preserving_V}. The corresponding transformed control is
\begin{align}
\begin{split}
\label{cnot_time_reversal_flip}
\widetilde{\bm f}^{(X_2)}(t) &=
\bigl( -f_{x,1}(T-t), -f_{y,1}(T-t), \\
&\qquad -f_{x,2}(T-t), +f_{y,2}(T-t)\bigr).
\end{split}
\end{align}
The minus sign together with conjugation by $X_2$ leaves the drift $Z_1Z_2$ invariant and produces the sign pattern in Eq.~\eqref{cnot_time_reversal_flip}. Thus $U_{\widetilde{\bm f}^{(X_2)}}(T) = X_2 U_{\bm f}(T)^\dagger X_2$. Consider, for example, a residual error component $E\propto Z_1X_2$, which is unchanged by the target-preserving twirl. Since {this error satisfies Eq.~\eqref{time_reversal_commuting_condition}, the time-reversal pairing cancels it to first order.}

The two constructions can be concatenated by applying the time-reversal pairing to each target-preserving-twirled branch. For $g\in\{I,Z_1\}$, define $H_{\widetilde{\bm f}^{(g,X_2)}}(t) = -X_2H_{\bm f^{(g)}}(T-t)X_2$. The concatenated protocol applies the four branches
\begin{align}
\label{concatenated_control_protocol_CNOT}
\bm f^{(I)}, \quad \bm f^{(Z_1)}, \quad \widetilde{\bm f}^{(I,X_2)}, \quad \widetilde{\bm f}^{(Z_1,X_2)}
\end{align}
with probability $1/4$ each. Thus the concatenated protocol eliminates all errors containing $X_1$ or $Y_1$, together with the residual errors $X_2$, $Z_1$, and $Z_1X_2$. Only ${\rm span}\{Y_2,Z_2,Z_1Y_2,Z_1Z_2\}$ can survive at first order.

\begin{figure}
    \centering
    \includegraphics[width=8.6cm]{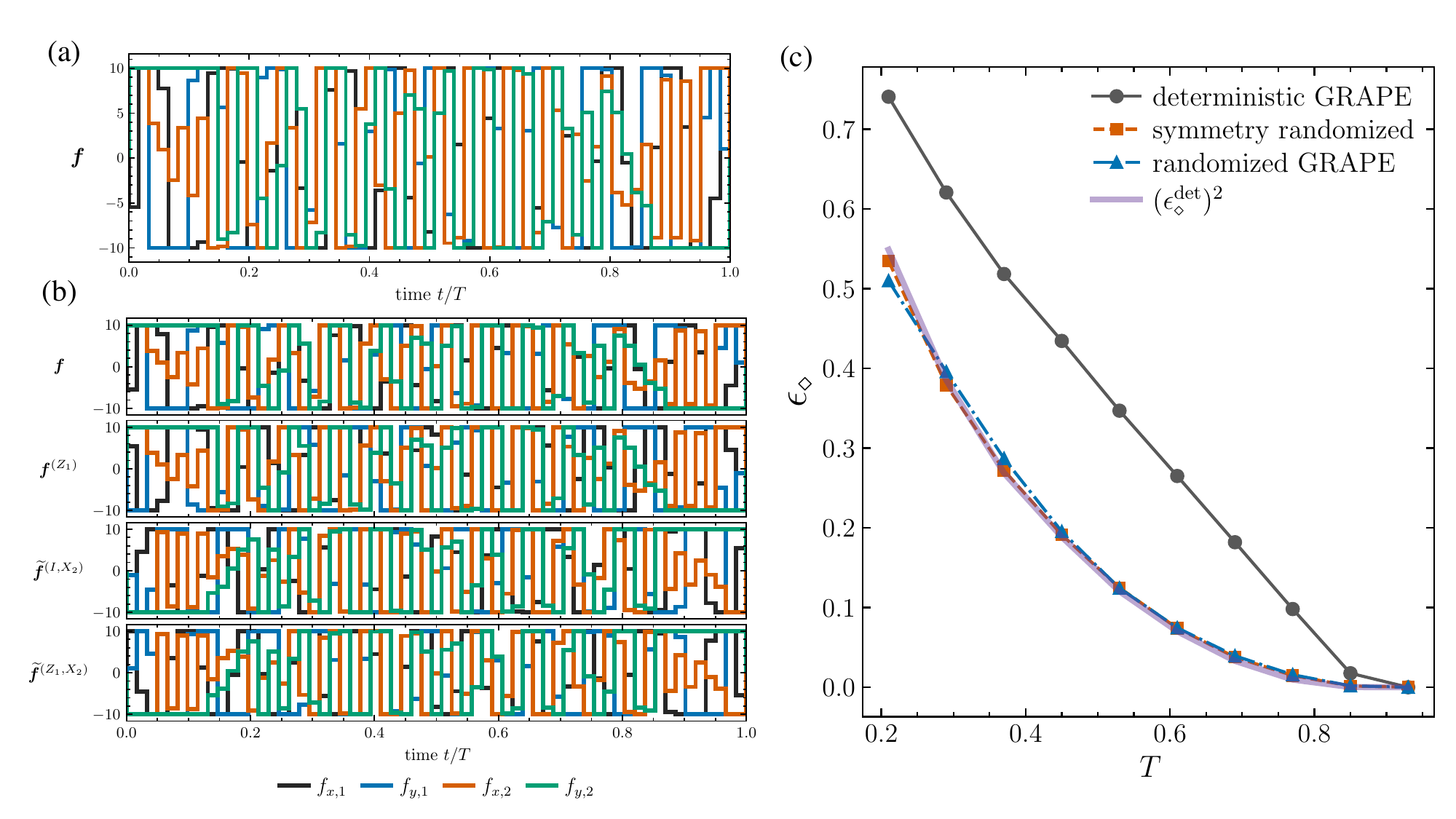}
    \caption{Optimal control of CNOT implementation in the two-qubit model of \eqref{two_qubit_model}, with $J=1$ and $f_{\max}=10$. (a) Deterministic GRAPE solution $\bm f$. (b) Four symmetry-generated branches of \eqref{concatenated_control_protocol_CNOT}, obtained by sign flips and time reversal of $\bm f$, and applied with equal probabilities of $1/4$. (c) Diamond-distance error versus total time $T$. The symmetry-generated protocol suppresses the deterministic GRAPE error quadratically, closely following $[\epsilon_\diamond^{\rm det}(T)]^2$ (faint purple guide), and achieves essentially the same performance as randomized GRAPE.}
    \label{fig:randomized-cnot}
\end{figure}

\textit{{CNOT benchmark---}}
We benchmark the symmetry-generated randomized protocols for CNOT implementation in the two-qubit model {of Eq.}~\eqref{two_qubit_model}, with $J=1$ and $f_{\max}=10$. For each $T$, we first obtain a reference pulse $\bm f$ using a deterministic GRAPE optimization \cite{khaneja2005optimal}. A representative deterministic solution is shown in Fig.~\ref{fig:randomized-cnot}(a).

From this waveform, we generate the four randomized branches in {Eq.}~\eqref{concatenated_control_protocol_CNOT} using the target-preserving twirl and time-reversal pairing described above. Explicitly, the randomized ensemble is obtained by applying sign flips and time reversal to the deterministic pulse, {see} Fig.~\ref{fig:randomized-cnot}(b). For comparison, we also develop a randomized GRAPE algorithm that directly optimizes over control ensembles and probabilities 
(see the SM).

Fig.~\ref{fig:randomized-cnot}(c) shows the diamond-distance error to the target CNOT channel. 
Denoting the deterministic GRAPE error by $\epsilon_\diamond^{\rm det}(T)$, the symmetry-generated randomized protocol closely follows $\left(\epsilon_\diamond^{\rm det}(T)\right)^2$, showing the same quadratic suppression predicted by the exactly solvable single-qubit example. {This behavior is consistent with cancellation of the leading coherent error.}

The {reason} is revealed by the deterministic relative error, whose leading component lies predominantly in the $Z_1X_2$ error. As discussed above, this error is unaffected by the target-preserving twirl, but is canceled by the time-reversal component of the symmetry protocol. The observed quadratic scaling therefore follows directly from coherent-error cancellation in the randomized ensemble. The symmetry-generated protocol is nearly indistinguishable from randomized GRAPE over the plotted range, suggesting that the symmetry-based construction captures the dominant mechanism underlying the randomized advantage in this setting. Further numerical details and error diagnostics are provided in the SM.

\textit{Robustness to noise---}
The error cancellation mechanism underlying randomized QOC can suppress coherent noise during {target gate implementation}. Consider {a setting where} the controlled dynamics $H_{\bm f}(t)$ of {Eq.}~\eqref{model} {is modified by} weak coherent noise,
\begin{align}
\widetilde{H}_{\bm f}(t):=H_{\bm f}(t)+\eta V, \qquad \eta\ll1,
\end{align}
where $V=V^\dagger$ and $\Tr(V)=0$. For simplicity, {assume} that a reference control $\bm f$ implements the target exactly, $U_{\bm f}(T)=U_{\rm targ}$. Then the perturbed evolution {is, to first order,} $U_\eta(T) = U_{\rm targ}e^{-i\eta F_V+\mathcal O(\eta^2)}$,
where $F_V = \int_0^T U_{\bm f}(t)^\dagger V U_{\bm f}(t)dt$.

{For this particular noise-robust protocol}, we construct the randomization through boundary pulses. Let $\{B_\alpha,p_\alpha\}$ be an ensemble of implementable unitaries. For each branch, we apply $B_\alpha$ before the reference control and
\begin{align}
C_\alpha := U_{\rm targ}B_\alpha^\dagger U_{\rm targ}^\dagger
\end{align}
after it. In the absence of noise, every branch yields the same target operation, $C_\alpha U_{\rm targ}B_\alpha = U_{\rm targ}$. Assuming for the moment instantaneous boundary pulses, the noisy evolution in branch $\alpha$ becomes $C_\alpha U_\eta(T) B_\alpha = U_{\rm targ}
\exp\left(-i\eta B_\alpha^\dagger F_V B_\alpha + \mathcal O(\eta^2)\right)$.
Averaging over the ensemble therefore replaces the coherent-noise generator $F_V$ by its averaged counterpart, $\overline F_V := \sum_\alpha p_\alpha B_\alpha^\dagger F_V B_\alpha$, yielding
\begin{align}
\sum_\alpha p_\alpha C_\alpha U_\eta(T) B_\alpha = U_{\rm targ}\left(I-i\eta\overline F_V + \mathcal O(\eta^2) \right).
\end{align}
{Thus, the present construction replaces the coherent-noise generator $F_V$ by $\overline F_V$, as randomized QOC replaces the coherent control error generator $E$ by its ensemble average $\overline E$.} When the ensemble is chosen so that $\overline F_V=0$, the first-order coherent-noise contribution is canceled and the error scales as $\mathcal O(\eta^2)$. This procedure may be viewed as a generalization of randomized dynamical decoupling \cite{yi2026faster} from identity to nontrivial target operations.

As a concrete example, consider any reference waveform that implements $U_{\rm targ} = {\rm CNOT}_{1\to2}$. exactly in the absence of noise. Choosing $B_\alpha$ uniformly from $\{I,X,Y,Z\}^{\otimes2}$ gives $\overline F_V = \frac1{16} \sum_B B^\dagger F_V B = 0$, so that the first-order coherent-noise contribution is completely removed. Since CNOT is {a} Clifford {gate}, each $C_\alpha = U_{\rm targ}B_\alpha^\dagger U_{\rm targ}^\dagger$ is again a Pauli string, e.g., $X_1\mapsto X_1X_2$, $Z_1\mapsto Z_1$, etc. Thus all required $B_\alpha$ and $C_\alpha$ remain experimentally simple.

Fig.~\ref{fig:boundary_noise_robustness_CNOT} illustrates the resulting noise-robust randomized CNOT protocol for coherent 1-local noise, $V = \sum_{j=1}^{2} \sum_{P\in\{X,Y,Z\}} v_{j,P}P_j,~ v_{j,P}\in[0,1]$. The randomized protocol exhibits the predicted $\mathcal O(\eta^2)$ diamond-distance scaling.

In practice, boundary pulses have finite duration $\tau_p$, which introduces an $\mathcal O(\eta\tau_p)$ contribution from noise acting during the pulses. However, this {linear} contribution can also be averaged {out}. In the SM, we construct randomized finite-width pulse sequences that implement Pauli operators with average error $\mathcal O(\eta^2\tau_p^2)$, thereby preserving the quadratic suppression in realistic implementations.
\begin{figure}
    \centering
    \includegraphics[width=7.6cm]{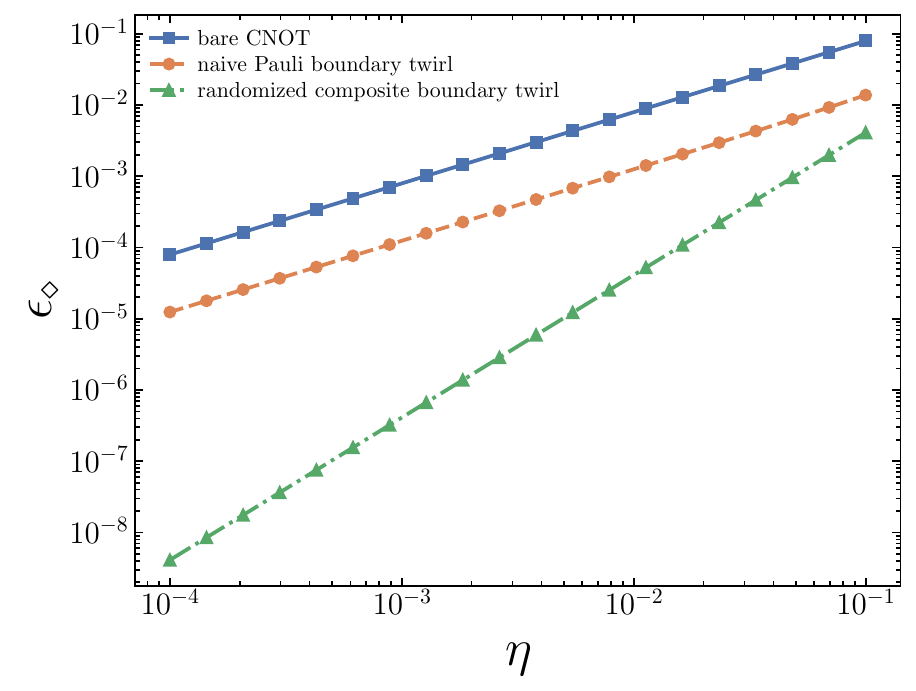}
    \caption{Noise robustness of CNOT implementation under coherent 1-local noise. The randomized composite boundary twirl suppresses the leading coherent-noise contribution and achieves $\mathcal{O}(\eta^2)$ diamond-distance scaling, whereas the bare CNOT and naive Pauli boundary-twirl protocols exhibit $\mathcal{O}(\eta)$ scaling.}
    \label{fig:boundary_noise_robustness_CNOT}
\end{figure}

\textit{Conclusion---}
{We showed that randomized quantum optimal control can outperform deterministic optimal control under identical resource constraints. In an exactly solvable single-qubit example, the optimal randomized error equals the square of the optimal deterministic error, implying a strict reduction in the minimum time required to achieve a given target accuracy.}

{The advantage relies in coherent-error cancellation between different control branches. Motivated by this principle, we developed symmetry-based constructions that generate randomized ensembles from deterministic controls and suppress leading-order errors. For CNOT implementation, these constructions reproduce the quadratic error suppression predicted by the solvable example and achieve performance comparable to direct randomized optimization. The same coherent-error-cancellation can be used to build noise-robust control protocols: randomized boundary-pulse schemes cancel leading-order coherent 1-local noise for Clifford targets.}

This work opens many avenues for future research, including extending the framework beyond mixed-unitary channels, establishing bounds on the power of randomized techniques, and tailoring these methods to hardware-specific gate implementations, particularly in regimes where coherent or biased noise dominates



\textit{Acknowledgments---}This work is support by the DOE Early Career Research Award No. DE-SC0026373. Additional support by DOE Express Award Number DE-SC0024685 is acknowledged.

\bibliography{prl_ref}

\clearpage
\onecolumngrid

\begin{center}
\textbf{\large Supplemental Material}
\end{center}

\setcounter{table}{0}
\renewcommand{\thetable}{S\arabic{table}}
\setcounter{figure}{0}
\renewcommand{\thefigure}{S\arabic{figure}}
\setcounter{equation}{0}
\renewcommand{\theequation}{S\arabic{equation}}
\setcounter{secnumdepth}{3}
\setcounter{theorem}{0}
\renewcommand{\thetheorem}{S\arabic{theorem}}
\setcounter{lemma}{0}
\renewcommand{\thelemma}{S\arabic{lemma}}
\setcounter{proposition}{0}
\renewcommand{\theproposition}{S\arabic{proposition}}

\section{Deterministic optimum for the single-qubit example}

In this section, we prove the deterministic optimal control result used in the single-qubit example in the main text. We begin by stating the precise problem statement. 

\begin{problem}[Deterministic optimal control problem]
\label{prob:sm_det_oc}
Consider the controlled single-qubit Hamiltonian
\begin{align}
    H_f(t)=Z+f(t)X,\qquad |f(t)|\le f_{\max},
\end{align}
and the target unitary channel
\begin{align}
    U_{\rm targ}=e^{-i\pi Y/2}, \qquad \mathcal U_{\rm targ}(\rho) := U_{\rm targ}\rho U_{\rm targ}^\dagger .
\end{align}
For an admissible control $f:[0,T]\to[-f_{\max},f_{\max}]$, let
\begin{align}
    U_f(T):= \mathcal T\exp\left(-i\int_0^T H_f(t)dt
    \right), \qquad \mathcal U_f(\rho) := U_f(T)\rho U_f(T)^\dagger .
\end{align}
For given fixed $T$ and $f_{\max}$, the deterministic optimal control problem is
\begin{align}
    \epsilon_{\diamond}^{\rm det,\star}(T) := \min_{f: |f(t)|\le f_{\max}}
    \frac12 \left\| \mathcal U_f-\mathcal U_{\rm targ}\right\|_\diamond.
    \label{sm_det_oc_problem}
\end{align}
\end{problem}

We now state the main theorem of this section.

\begin{theorem}[Deterministic optimal control solution] \label{thm:sm_det_opt}
Assume that $0<f_{\max}\le 1$ and $T\sqrt{1+f_{\max}^2} < \pi$. Then the following symmetric two-bang control
\begin{align}
    f_+(t)=
    \begin{cases}
        +f_{\max}, & 0\le t<T/2,\\
        -f_{\max}, & T/2\le t\le T
    \end{cases}
    \label{sm_f_plus_det}
\end{align}
solves Problem~\ref{prob:sm_det_oc}. The corresponding minimum error is
\begin{align}
\epsilon_{\diamond}^{\rm det,\star}(T) = \sqrt{ 1- \left( \frac{2f_{\max}}{1+f_{\max}^2} \sin^2\left(\frac{T\sqrt{1+f_{\max}^2}}{2} \right)\right)^2}.
\end{align}

\end{theorem}

The proof proceeds in three steps. First, we rewrite Problem~\ref{prob:sm_det_oc} as the problem of maximizing the $Y$-component $|y(T)|$ of the final unitary, and then use a symmetry of the control system to reduce this further to maximizing $y(T)$ (Section~\ref{subsec:dia-to-Y}). Second, we apply the Pontryagin maximum principle to characterize the possible optimal extremals. In particular, we show that an optimal extremal can be chosen from a finite family of bang-bang or bang-off-bang controls (Section~\ref{subsec:pmp-reduction}).  Finally, we evaluate these candidate families explicitly and show the two-bang protocol \eqref{sm_f_plus_det} maximizes $y(T)$, thus proving Theorem~\ref{thm:sm_det_opt} (Section~\ref{subsec:candidate-comparison}).

\subsection{From diamond distance to the $Y$-component objective}\label{subsec:dia-to-Y}

In this subsection, we reduce Problem~\ref{prob:sm_det_oc} to the problem of maximizing the final $Y$-component of the final unitary. 

We first note the following elementary form of the diamond distance between unitary channels for $SU(2)$. 

\begin{proposition}[Diamond distance for $SU(2)$ unitary channels]
\label{prop:sm_su2_diamond_distance}
Let $U,V \in SU(2)$, and let
\begin{align}
\mathcal U(\rho):=U\rho U^\dagger,\qquad \mathcal V(\rho):=V\rho V^\dagger .
\end{align}
Then
\begin{align}
\frac12 \left\| \mathcal U-\mathcal V \right\|_\diamond = \sqrt{1- \left| \frac12\Tr(U^\dagger V)\right|^2}.
\label{sm_su2_diamond_distance}
\end{align}
\end{proposition}

\begin{proof}
For two unitary channels, Theorem~3.55 of Ref.~\cite{watrous2018the}, with $\lambda=1/2$ and with the isometries specialized to unitaries, implies
\begin{align}
\label{sm_unitary_channel_distance}
    \frac12 \left\|\mathcal U-\mathcal V\right\|_\diamond = \sqrt{1-\delta(U^\dagger V)^2}, \qquad \delta(W):= \min_{\|u\|=1}|u^\dagger W u|.
\end{align}
Let $W=U^\dagger V\in SU(2)$.  Every element of $SU(2)$ can be written as
\begin{align}
    W=\alpha I-i\bm\beta\cdot\bm\sigma, \qquad \alpha\in\mathbb R,\quad \bm\beta\in\mathbb R^3,\quad \alpha^2+\|\bm\beta\|^2=1.
\end{align}
For a pure state with Bloch vector $\bm r$, $\|\bm r\|=1$, we have $    \bra{\psi} W \ket{\psi} = \alpha-i\bm\beta\cdot\bm r$. Therefore $|\bra{\psi} W \ket{\psi}|^2 = \alpha^2+(\bm\beta\cdot\bm r)^2 \ge \alpha^2$, where equality is achieved by choosing $\bm r\perp \bm\beta$, which is always possible. Hence, $\delta(W)=|\alpha|$, and since $\alpha=\frac12\Tr(W)=\frac12\Tr(U^\dagger V)$, we obtain \eqref{sm_su2_diamond_distance}.
\end{proof}

We now apply Proposition~\ref{prop:sm_su2_diamond_distance} to the target $U_{\rm targ}=e^{-i\pi Y/2}=-iY$. 

\begin{lemma}[Reduction to the $Y$-component]
\label{lem:sm_reduction_to_y}
Let $U_f(T) = a_f(T)I-i\bigl(x_f(T)X+y_f(T)Y+z_f(T)Z\bigr)$. Then
\begin{align}
\label{sm_diamond_y_relation}
\frac12 \left\|\mathcal U_f-  \mathcal U_{\rm targ}\right\|_\diamond = \sqrt{1-y_f(T)^2}.
\end{align}
Consequently, Problem~\ref{prob:sm_det_oc} is equivalent to maximizing $|y_f(T)|$ over all admissible controls.
\end{lemma}

\begin{proof}
Since $U_{\rm targ}=e^{-i\pi Y/2}=-iY$, we have
\begin{align}
\frac12\Tr\left(U_{\rm targ}^\dagger U_f(T)\right) = \frac12\Tr\left(iY\left(a_f I-i(x_fX+y_fY+z_fZ)\right)\right) = y_f(T).
\end{align}
The claim follows directly from Proposition~\ref{prop:sm_su2_diamond_distance}.
\end{proof}

The preceding lemma reduces Problem~\ref{prob:sm_det_oc} to the problem of maximizing $|y_f(T)|$. We next use a symmetry of the dynamics to remove the absolute value. 

\begin{lemma}[Sign symmetry] \label{lem:sm_sign_symmetry}
For the control system in Problem~\ref{prob:sm_det_oc}, the maximization of $|y_f(T)|$ is equivalent to the maximization of $y_f(T)$:
\begin{align} \label{sm_abs_y_to_y}
\max_{f: |f(t)| \le f_{\max}} |y_f(T)|= \max_{f: |f(t)| \le f_{\max}} y_f(T).
\end{align}
\end{lemma}
\begin{proof}
If $f$ is admissible, then $-f$ is also admissible. Moreover, $H_{-f}(t) =Z-f(t)X=ZH_f(t)Z$. Let $U_f(t)$ be the unitary evolution generated by $H_f(t)$, then $U_{-f}(T)=ZU_f(T)Z$. Therefore, if $U_f(T) = a_f I-i(x_fX+y_fY+z_fZ)$, then, using $ZXZ=-X$, $ZYZ=-Y$, and $ZZZ=Z$, we obtain $U_{-f}(T) = a_f I-i(-x_fX-y_fY+z_fZ)$. Thus replacing $f$ by $-f$ sends $y_f(T)$ to $-y_f(T)$. Hence every achievable value of $y_f(T)$ is accompanied by its negative, and therefore maximizing $|y_f(T)|$ is equivalent to maximizing $y_f(T)$.
\end{proof}

We also show that the short-time condition itself rules out exact reachability of the target channel.

\begin{lemma}[No exact implementation in the short-time regime] \label{lem:sm_no_exact_short_time} 
Set $\Omega:=\sqrt{1+f_{\max}^2}$. If $\Omega T<\pi$, then no admissible control satisfying $|f(t)|\le f_{\max}$ exactly implements the target channel $\mathcal U_{\rm targ}(\rho) = Y \rho Y$. 
\end{lemma}

\begin{proof} 
Write $U_f(t)=a(t)I-i\bigl(x(t)X+y(t)Y+z(t)Z\bigr)$ and define \begin{align} 
A(t):=a(t)+iy(t),\qquad B(t):=x(t)+iz(t). 
\end{align} 
The Schr\"odinger equation with $H_f(t)=Z+f(t)X$ gives 
\begin{align} 
\dot A=(i-f)B,\qquad \dot B=(i+f)A. 
\end{align} 
Moreover, unitarity gives $|A(t)|^2+|B(t)|^2=1$, and the initial condition $U_f(0)=I$ gives $A(0)=1$ and $B(0)=0$.

Now define 
\begin{align}
W(t):=A(t)^*B(t),\qquad D(t):=|A(t)|^2-|B(t)|^2. 
\end{align} 
Using the equations for $A$ and $B$, we obtain 
\begin{align} 
\dot W = \dot A^*B+A^*\dot B = (-i-f)|B|^2+(i+f)|A|^2 = (i+f)D. 
\end{align}

Now assume that there exists an admissible control $f$ that implements the target channel exactly at time $T$. Then, since $B(0)=B(T)=0$, we have $W(0)=W(T)=0$. Therefore 
\begin{align} 
0 = W(T)-W(0) = \int_0^T (i+f(t))D(t) dt. 
\end{align} 
Taking imaginary parts gives 
\begin{align} 
\int_0^T D(t) dt=0. 
\end{align} 
Together with $|A(t)|^2+|B(t)|^2=1$, this implies 
\begin{align} \int_0^T |A(t)|^2 dt = \int_0^T |B(t)|^2 dt = \frac{T}{2}.
\end{align}
Since $B(0)=B(T)=0$, the Wirtinger inequality \cite{hardy1952inequalities} gives 
\begin{align} 
\int_0^T |\dot B(t)|^2 dt \ge \frac{\pi^2}{T^2} \int_0^T |B(t)|^2 dt = \frac{\pi^2}{2T}. 
\end{align} 
On the other hand, 
\begin{align} 
|\dot B(t)|^2 = |i+f(t)|^2|A(t)|^2 = (1+f(t)^2)|A(t)|^2 \le \Omega^2 |A(t)|^2. 
\end{align}
Hence 
\begin{align} 
\int_0^T |\dot B(t)|^2 dt \le \Omega^2 \int_0^T |A(t)|^2 dt = \frac{\Omega^2T}{2}. 
\end{align} 
Combining the two bounds yields
\begin{align} 
\frac{\pi^2}{2T} \le \frac{\Omega^2T}{2}, 
\end{align} 
and therefore $\Omega T\ge \pi$. This contradicts the assumption $\Omega T<\pi$. Hence exact implementation is impossible in the short-time regime.
\end{proof}

\subsection{Pontryagin maximum principle and structure of extremal solutions} \label{subsec:pmp-reduction}
Combining Lemmas~\ref{lem:sm_reduction_to_y} and \ref{lem:sm_sign_symmetry}, it remains to solve the maximization problem
\begin{align}
\label{sm_y_max_problem}
    \max_{f: |f(t)|\le f_{\max}} y_f(T).
\end{align}
We now characterize the possible optimal extremals for \eqref{sm_y_max_problem} using the Pontryagin maximum principle.

To do so, we first denote the unitary evolution as
\begin{align}
U_f(t) = a(t)I-i\bigl(x(t)X+y(t)Y+z(t)Z\bigr),
\end{align}
and define
\begin{align}
q(t) := (a(t),x(t),y(t),z(t))^T.
\end{align}
In these coordinates, left multiplication by $-iX$, $-iY$, and $-iZ$ is represented by the real skew-symmetric matrices
\begin{align}
A_X &= \begin{pmatrix}
0&-1&0&0\\
1&0&0&0\\
0&0&0&-1\\
0&0&1&0
\end{pmatrix},&
A_Y &= \begin{pmatrix}
0&0&-1&0\\
0&0&0&1\\
1&0&0&0\\
0&-1&0&0
\end{pmatrix},&
A_Z &=
\begin{pmatrix}
0&0&0&-1\\
0&0&-1&0\\
0&1&0&0\\
1&0&0&0
\end{pmatrix}.
\end{align}
Equivalently,
\begin{align}
\label{sm_AX_AZ_def}
A_Xq=(-x,a,-z,y)^T,\qquad
A_Yq=(-y,z,a,-x)^T,\qquad
A_Zq=(-z,-y,x,a)^T .
\end{align}
The Schr\"odinger equation of the single-qubit system is therefore equivalent to
\begin{align}
\label{sm_q_dynamics}
\dot q(t)=\bigl(A_Z+f(t)A_X\bigr)q(t),\qquad
q(0)=e_a:=(1,0,0,0)^T .
\end{align}
The terminal cost that we aim to maximize is
\begin{align}
y(T)=e_y^T q(T), \qquad e_y:=(0,0,1,0)^T .
\end{align}

We now state the Pontryagin maximum principle (PMP) conditions for the problem in \eqref{sm_y_max_problem}, which provide necessary conditions for a control to be extremal. A detailed review of PMP is beyond the scope of this work; we refer the reader to standard treatments, e.g., Refs.~\cite{boscain2021introduction,ansel2024introduction}.

\begin{lemma}[Pontryagin maximum principle conditions]
\label{lem:sm_pmp_setup}
Let $f^\star$ be an optimal control for \eqref{sm_y_max_problem}, with corresponding trajectory $q^\star(t)$. Then there exists a costate $p(t)\in\mathbb R^4$ such that
\begin{align}
\label{sm_adjoint_equation}
\dot p(t) = \bigl(A_Z+f^\star(t)A_X\bigr)p(t), \qquad p(T)=e_y.
\end{align}
Moreover, for almost every $t\in[0,T]$, $f^\star(t)$ maximizes the Pontryagin Hamiltonian
\begin{align}
\label{pontryagin-hamiltonian}
\mathscr H(q,p,f) := p^T(A_Z+fA_X)q
\end{align}
over $f\in[-f_{\max},f_{\max}]$.
\end{lemma}

\begin{proof}
We apply the Pontryagin maximum principle to the fixed-time maximization problem \eqref{sm_y_max_problem} in $\mathbb R^4$. Since $A_X,A_Z$ are skew-symmetric, the adjoint equation $\dot p=-\partial_q\mathscr H$ gives \eqref{sm_adjoint_equation}. The terminal cost is $e_y^Tq(T)$, so the transversality condition gives $p(T)=-\lambda_0 e_y$ for some $\lambda_0 \le 0$. If $\lambda_0=0$, then $p(T)=0$, and the homogeneous adjoint equation implies $p(t) \equiv 0$. This contradicts the nontriviality condition of the Pontryagin maximum principle, which states that the PMP multipliers cannot vanish simultaneously, i.e., $(p(t),\lambda_0)\neq(0,0)$ for all $t$. Hence $\lambda_0 < 0$. After rescaling the costate, we set $\lambda_0=-1$, and therefore $p(T)=e_y$. Finally, the Hamiltonian maximization condition is the maximum condition in the
Pontryagin maximum principle.
\end{proof}

We now introduce the switching function. For notational simplicity, in the rest of this subsection we write $q(t)=q^\star(t)$ and $f(t)=f^\star(t)$ along a fixed extremal. Define
\begin{align}
\phi(t):=p(t)^TA_Xq(t), \qquad \chi(t):=p(t)^TA_Zq(t),
\end{align}
where $\phi(t)$ is denoted by the switching function. Then the Pontryagin Hamiltonian along the extremal is
\begin{align}
\mathscr H(q(t),p(t),f) = \chi(t)+f\phi(t),
\end{align}
and the maximum condition in Lemma~\ref{lem:sm_pmp_setup} therefore gives
\begin{align}
\label{sm_phi_bang_condition}
f(t) = \begin{cases}
        +f_{\max}, & \phi(t)>0,\\
        -f_{\max}, & \phi(t)<0,
    \end{cases}
\end{align}
for almost every $t$ such that $\phi(t) \neq 0$. Thus zeros of $\phi$ are the only times at which the PMP does not immediately determine the control.

To analyze such zeros, we first analyze the differential equations satisfied by the switching functions.

\begin{proposition}[Switching-function equations] \label{lem:sm_switching_equations}
Define $\eta(t):=p(t)^TA_Yq(t)$ where $A_Yq=(-y,z,a,-x)^T$. Then
\begin{align}
\begin{split}
\label{sm_switching_equations}
\dot\phi(t)&=-2\eta(t),\\
\dot\eta(t)&=2\phi(t)-2f(t)\chi(t),\\
\dot\chi(t)&=2f(t)\eta(t).
\end{split}
\end{align}
\end{proposition}

\begin{proof}
Let $A_f:=A_Z+f(t)A_X$. Since $\dot q=A_fq$ and $\dot p=A_fp$, for any $B\in\{A_X,A_Y,A_Z\}$ we have
\begin{align}
\frac{d}{dt}\left(p^TBq\right) = p^T\left(A_f^TB+BA_f\right)q.
\end{align}
Using 
\begin{align}
A_XA_Z-A_ZA_X = -2A_Y, \quad A_YA_Z-A_ZA_Y = 2A_X, \quad A_YA_X-A_XA_Y = -2A_Z,
\end{align}
we obtain
\begin{align}
\dot\phi(t) &= p^T(A_XA_Z - A_ZA_X)q = -2\eta(t),\\
\dot\eta(t) &= p^T(A_YA_Z - A_ZA_Y)q + f(t)p^T(A_YA_X - A_XA_Y)q \\ 
&= 2\phi(t)-2f(t)\chi(t),\\
\dot\chi(t) &= f(t)p^T(A_ZA_X - A_XA_Z)q = 2f(t)\eta(t).
\end{align}
This proves the claim.
\end{proof}

Then, the following lemma shows that on any singular region, that is, any interval on which $\phi(t)=0$, the control function must be $0$.

\begin{lemma}[Singular regions are off controls] \label{lem:sm_singular_regions_off}
Assume that the optimal final unitary is not an exact implementation of the target, so that $|y(T)|<1$ for the maximization problem \eqref{sm_y_max_problem}. If $\phi(t)\equiv 0$ on an open interval $I\subset[0,T]$, then
\begin{align}
    f(t)=0
\end{align}
for almost every $t\in I$.
\end{lemma}

\begin{proof}
Suppose $\phi(t)\equiv 0$ on $I$. From Proposition~\ref{lem:sm_switching_equations} we have $\dot\phi(t)=-2\eta(t)$, and therefore $\eta(t)\equiv 0$ on $I$. Using it again gives $0 = \dot\eta(t) = 2\phi(t)-2f(t)\chi(t) = -2f(t)\chi(t)$ for almost every $t\in I$. Hence $f(t)\chi(t)=0$ for almost every $t\in I$. 

We now claim that $\chi(t)$ cannot vanish at any point of $I$ in the non-exact case. Indeed, if $\chi(t_0)=0$ for some $t_0\in I$, the earlier equations give $\phi(t_0) = \eta(t_0) = \chi(t_0)=0$. Equivalently,
\begin{align}
p(t_0)^TA_Xq(t_0) = p(t_0)^TA_Yq(t_0) = p(t_0)^TA_Zq(t_0) = 0.
\end{align}
Since $q(t_0)$ has unit norm, the vectors $A_Xq(t_0)$, $A_Yq(t_0)$, and $A_Zq(t_0)$ form an orthonormal basis of $q(t_0)^\perp$; indeed, this follows by a direct calculation from \eqref{sm_AX_AZ_def} 
since they are all orthogonal to $q(t_0)$, pairwise orthogonal to one another, and have norm $\|q(t_0)\|=1$. Therefore the above identities imply that $p(t_0)$ is orthogonal to $q(t_0)^\perp$, and hence
\begin{align}
    p(t_0)=\lambda q(t_0)
\end{align}
for some $\lambda\in\mathbb R$. Since $A_X$ and $A_Z$ are skew-symmetric, $A(t):=A_Z+f(t)A_X$ is skew-symmetric for every $t$. Hence, for any solution $r$ of $\dot r=A(t)r$,
\begin{align}
\frac{d}{dt}\|r(t)\|_2^2 = r(t)^T(A(t)^T+A(t))r(t)=0.
\end{align}
Thus the Euclidean norms of both $q$ and $p$ are conserved. Since $q(0)=e_a$ and $p(T)=e_y$, we have $\|q(t_0)\|_2=1$ and $\|p(t_0)\|_2=\|p(T)\|_2=\|e_y\|_2=1$. Therefore, from $p(t_0)=\lambda q(t_0)$, we obtain $1=\|p(t_0)\|_2=|\lambda| \|q(t_0)\|_2=|\lambda|$. Hence $\lambda=\pm1$.

Because $p$ and $\lambda q$ obey the same linear differential equation along the extremal and agree at $t_0$, uniqueness gives $p(t)=\lambda q(t)$ for all $t$. Evaluating at $t=T$ yields
\begin{align}
    e_y=p(T)=\lambda q(T),
\end{align}
so $q(T)=\lambda e_y$ and hence $|y(T)|=1$. This is an exact implementation of the target channel, contradicting Lemma~\ref{lem:sm_no_exact_short_time}. Therefore $\chi(t)$ cannot vanish on $I$. Since $f(t)\chi(t) = 0$ almost everywhere on $I$, we conclude that $f(t) = 0$ almost everywhere on $I$.
\end{proof}

We next state a simple time-reversal identity that will be used below. This is also used in the time-reversal pairing construction in \eqref{endpoint_time_reversal} in the main text.

\begin{proposition}[Time-reversal identity]
\label{prop:sm_time_reversal_identity}
Let $U(t)$ be generated by a time-dependent Hamiltonian $H(t)$,
\begin{align}
\dot U(t)=-iH(t)U(t), \qquad U(0)=I.
\end{align}
Fix a unitary $v$, and let $\widetilde U(t)$ be generated by
\begin{align}
\widetilde H(t) := - v^\dagger H(T-t)v, \qquad \dot{\widetilde U}(t)=-i\widetilde H(t)\widetilde U(t), \qquad \widetilde U(0)=I .
\end{align}
Then
\begin{align}
    \widetilde U(T)=v^\dagger U(T)^\dagger v.
\end{align}
\end{proposition}

\begin{proof}
Define
\begin{align}
    V(t):=v^\dagger U(T-t)U(T)^\dagger v .
\end{align}
Then $V(0)=I$. Moreover,
\begin{align}
\dot V(t) &= v^\dagger \frac{d}{dt}U(T-t) U(T)^\dagger v \\
&= i v^\dagger H(T-t)U(T-t)U(T)^\dagger v \\
&= i v^\dagger H(T-t)vV(t) \\
&= -i\widetilde H(t)V(t).
\end{align}
Thus $V(t)$ satisfies the same Schr\"odinger equation and initial condition as $\widetilde U(t)$. By uniqueness, $V(t)=\widetilde U(t)$. Evaluating at $t=T$ gives
\begin{align}
\widetilde U(T) = V(T) = v^\dagger U(0)U(T)^\dagger v = v^\dagger U(T)^\dagger v.
\end{align}

\end{proof}

We next prove a useful consequence of the PMP Hamiltonian, which will be used to control the spacing between switching times shortly.

\begin{lemma}[Non-negativity of maximized Pontryagin Hamiltonian] 
\label{lem:sm_nonnegative_HP} 
Define the maximized Pontryagin Hamiltonian 
\begin{align} 
\label{sm_K_definition} 
\mathscr H_{\max} := \max_{|u|\le f_{\max}} \mathscr H(q(t),p(t),u) = \chi(t)+f_{\max}|\phi(t)|. 
\end{align} 
Then $\mathscr H_{\max}$ is independent of $t$. Moreover, among optimal controls for \eqref{sm_y_max_problem}, one can choose an optimal extremal with 
\begin{align} 
\mathscr H_{\max} \ge 0. 
\end{align} 
\end{lemma}

\begin{proof}
We first prove that $\mathscr H_{\max}$ is constant. On a bang region where $\phi(t)>0$, the maximum condition gives $f(t)=f_{\max}$. Hence $\mathscr H_{\max}=\chi+f_{\max}\phi$. Using Proposition~\ref{lem:sm_switching_equations},
\begin{align}
\frac{d}{dt}\mathscr H_{\max} = \dot\chi+f_{\max}\dot\phi = 2f_{\max}\eta-2f_{\max}\eta = 0.
\end{align}
Similarly, on a bang region where $\phi(t)<0$, the maximum condition gives $f(t)=-f_{\max}$ and $\mathscr H_{\max}=\chi-f_{\max}\phi$. Thus
\begin{align}
\frac{d}{dt}\mathscr H_{\max} = \dot\chi-f_{\max}\dot\phi = -2f_{\max}\eta+2f_{\max}\eta = 0.
\end{align}
On a singular region, $\phi \equiv 0$. Then $\dot\phi=-2\eta$ implies $\eta \equiv 0$, and hence $\mathscr H_{\max}=\chi$ and $\dot\chi=2f\eta=0$. Hence $\mathscr H_{\max}$ is constant on every bang or singular region, and thus is independent of $t$.

It remains to show that one can choose an optimal extremal with $\mathscr H_{\max} \ge 0$. Let $f$ be any optimal control, denote $U_f(T)=aI-i(xX+yY+zZ)$, and define another admissible control
\begin{align}
    g(t) := -f(T-t).
\end{align}
Then, $H_g(t)=Z+g(t)X = Z-f(T-t)X = -XH_f(T-t)X$. By Proposition~\ref{prop:sm_time_reversal_identity}, with $v=X$, we obtain $U_g(T)=XU_f(T)^\dagger X$. Since $U_f(T)^\dagger=aI+i(xX+yY+zZ)$, we get $U_g(T) = aI+i(xX-yY-zZ) = aI-i(-xX+yY+zZ)$. Thus $g$ achieves the same final $Y$-component value $y(T)$ as $f$, while replacing $x(T)$ by $-x(T)$. Hence, among optimal controls, we may choose one with $x(T) \ge 0$.

For this chosen optimal extremal, using $p(T)=e_y$, we have $\phi(T) = p(T)^TA_Xq(T) = e_y^TA_Xq(T) = -z(T)$, and $\chi(T) = p(T)^TA_Zq(T) = e_y^TA_Zq(T) = x(T)$. Since $\mathscr H_{\max}$ is independent of $t$, it follows that
\begin{align}
\mathscr H_{\max} = \mathscr H_{\max}(T) = \chi(T)+f_{\max}|\phi(T)| = x(T)+f_{\max}|z(T)| \ge 0.
\end{align}
This proves the claim.
\end{proof}

We now use the nonnegative maximized Hamiltonian to control the spacing between zeros of the switching function.

\begin{lemma}[Minimum switching gap on bang regions] \label{lem:sm_switching_gap}
Consider an optimal extremal chosen so that $H_{\max}\ge0$, and set $\Omega:=\sqrt{1+f_{\max}^2}$. Suppose a bang interval of sign $s\in\{+1,-1\}$ leaves a switching time $t_0$, i.e. $\phi(t_0)=0,~ s\phi(t)>0$ for $t\in(t_0,t_1)$. If $t_1>t_0$ is the next zero of $\phi$, including the possibility that $t_1$ is the endpoint of the bang interval, then
\begin{align}
t_1-t_0\ge \frac{\pi}{2\Omega}.
\end{align}
In particular, any bang interval whose two endpoints are switching times has length at least $\pi/(2\Omega)$.
\end{lemma}

\begin{proof}
Denote the bang region as $(t_0,t_0+\ell)$, so that $f(t)=s f_{\max}$ for almost every $t\in(t_0,t_0+\ell)$. By the maximum condition, this implies $s\phi(t) = |\phi(t)|$ on the bang region. Hence, on this region, $\mathscr H_{\max} = \chi(t)+s f_{\max}\phi(t)$, and therefore $\chi(t) = \mathscr H_{\max} - s f_{\max}\phi(t)$. Using the switching-function equations in Proposition~\ref{lem:sm_switching_equations}, we obtain, for a.e. $t\in(t_0,t_0+\ell)$,
\begin{align}
\ddot\phi(t) &= -2\dot\eta(t) = -4\phi(t)+4s f_{\max}\chi(t) = -4(1+f_{\max}^2)\phi(t) + 4s f_{\max}\mathscr H_{\max}.
\end{align}
Thus, with $\Omega=\sqrt{1+f_{\max}^2}$,
\begin{align}
\label{sm_phi_bang_ode}
\ddot\phi(t)+4\Omega^2\phi(t) = 4s f_{\max}\mathscr H_{\max}.
\end{align}
Since $t_0$ is a zero of $\phi$, we have $\phi(t_0)=0$. Moreover, because the bang region leaves $t_0$ with sign $s$, we have $s\dot\phi(t_0) \ge 0$. Define $r:=-s\eta(t_0)$. Using $\dot\phi(t_0)=-2\eta(t_0)$, this gives
\begin{align}
    r \ge 0, \qquad \dot\phi(t_0)=2sr.
\end{align}
Solving \eqref{sm_phi_bang_ode} with initial conditions $\phi(t_0)=0$ and $\dot\phi(t_0)=2sr$ gives, for $0<\tau<\ell$,
\begin{align}
\label{sm_phi_bang_solution}
\phi(t_0+\tau) = \frac{s}{\Omega^2} \left(f_{\max}\mathscr H_{\max} \bigl(1-\cos(2\Omega\tau)\bigr) + r\Omega\sin(2\Omega\tau) \right).
\end{align}
Now let
\begin{align}
    0 < \tau < \min\left\{\ell,\frac{\pi}{2\Omega}\right\}.
\end{align}
Then $0<2\Omega\tau<\pi$, so $1 - \cos(2\Omega\tau) > 0$ and $\sin(2\Omega\tau) > 0$. Since $\mathscr H_{\max}\ge0$ from Lemma \ref{lem:sm_nonnegative_HP} and $r\ge0$, the bracket in \eqref{sm_phi_bang_solution} is strictly positive unless $\mathscr H_{\max} = r = 0$. But if $\mathscr H_{\max} = r = 0$, then \eqref{sm_phi_bang_solution} gives $\phi \equiv 0$ on the bang region, contradicting that the region is bang. Therefore, $s\phi(t_0+\tau) > 0$ for every $0 < \tau < \min\left\{\ell,\frac{\pi}{2\Omega}\right\}$. 

Now suppose $t_1 \in (t_0,t_0+\ell]$ is another zero of $\phi$ and $t_1 - t_0 < \pi/(2\Omega)$. If $t_1<t_0+\ell$, this contradicts the strict positivity just proved. If $t_1 = t_0+\ell$, then taking the limit $\tau\to\ell^-$ in \eqref{sm_phi_bang_solution} gives $s\phi(t_1)>0$, again contradicting $\phi(t_1)=0$. Hence any such zero must satisfy
\begin{align}
    t_1 - t_0 \ge \frac{\pi}{2\Omega}.
\end{align}

\end{proof}

We can now state the structure of optimal controls using the preceding lemmas.

\begin{proposition}[Structure of optimal controls] 
\label{prop:sm_extremal_structure} 
Assume $\Omega T < \pi$ and $\Omega=\sqrt{1+f_{\max}^2}$. Consider a non-exact optimal extremal for \eqref{sm_y_max_problem}, chosen so that $\mathscr H_{\max} \ge 0$. Then, the optimal control can be chosen from one of the following two families: 
\begin{enumerate} 
\item a pure bang-bang control with at most two switchings, 
\begin{align} 
f(t)\in\{+f_{\max},-f_{\max}\}, 
\end{align} 
\item a bang-off-bang control, 
\begin{align} f(t)= 
\begin{cases} 
s_1 f_{\max}, & 0\le t<t_1, \\ 
0, & t_1<t<t_2, \\ 
s_2 f_{\max}, & t_2<t\le T, 
\end{cases} \qquad 0\leq t_1 \leq t_2 \leq T, \qquad s_1,s_2\in\{+1,-1\}. 
\label{sm_bang_off_bang_form} 
\end{align} 

\end{enumerate} 
\end{proposition}

\begin{proof}
By the maximum condition in Lemma~\ref{lem:sm_pmp_setup}, every region on which $\phi(t) \neq 0$ is a bang region, with $f(t) = \pm f_{\max}$ almost everywhere. By Lemma~\ref{lem:sm_singular_regions_off}, every singular region, every open interval on which $\phi(t)=0$, is an off region with $f(t)=0$ almost everywhere.

We first state a consequence of the explicit bang-region solution used in Lemma~\ref{lem:sm_switching_gap}. Suppose a bang region of sign $s$ starts at a zero $t_0$ of $\phi$. For $\tau=t-t_0$, we have
\begin{align}
\label{sm_phi_bang_solution_AB}
\phi(t_0+\tau) = \frac{s}{\Omega^2} \left(A\bigl(1-\cos(2\Omega\tau)\bigr) + B\sin(2\Omega\tau) \right),
\end{align}
where $A:=f_{\max}\mathscr H_{\max} \ge 0$ and $B:=r\Omega \ge 0$. If, at some later time $t_0+L$ on the same bang region, both $\phi(t_0+L)=0$ and $\eta(t_0+L)=0$, then also $\dot\phi(t_0+L)=0$. Setting $u=2\Omega L$, these two conditions imply
\begin{align}
\begin{pmatrix}
    1-\cos u & \sin u \\
    \sin u & \cos u
\end{pmatrix}
\begin{pmatrix}
    A \\
    B
\end{pmatrix} = 0.
\end{align}
For $0 < u < 2\pi$, the determinant is $\cos u - 1\neq0$. Hence $A = B = 0$, which would make $\phi$ vanish throughout the bang region, a contradiction. Therefore a bang region cannot start at $\phi=0$ and later reach $\phi=\eta=0$ in time $L < \pi/\Omega$.

Now suppose first that there is no singular region. Then all zeros of $\phi$ in $(0,T)$ are isolated switching times. If there were three such zeros $t_1<t_2<t_3$, Lemma~\ref{lem:sm_switching_gap} would give
\begin{align}
    t_2 - t_1 \ge \frac{\pi}{2\Omega}, \qquad t_3 - t_2 \ge \frac{\pi}{2\Omega},
\end{align}
and hence $t_3-t_1 \ge \pi/\Omega$. This contradicts our assumption on the short-time regime $t_3-t_1 < T \le \pi/\Omega$. Thus, in the absence of singular regions, the extremal is bang-only with at most two switchings.

Now suppose that a singular region is present, and let $(\alpha,\beta)$ be a maximal singular region. Since $(\alpha,\beta)$ is a singular region, we have $\phi(t)=0$ for all $t \in (\alpha,\beta)$. The switching-function equation $\dot\phi=-2\eta$ then gives $\eta(t) = 0$ on $(\alpha,\beta)$. 
Since $q$ and $p$ solve linear differential equations with bounded measurable coefficients, they are continuous. Therefore $\phi(t)$ and $\eta(t)$ are continuous functions of $t$. Thus the identities $\phi=0$ and $\eta=0$, which hold on the singular interval $(\alpha,\beta)$, extend to any endpoint $\alpha$ or $\beta$ lying in the interior of $[0,T]$.

Since $\phi$ and $\eta$ are continuous, the same equalities hold at any endpoint $\alpha$ or $\beta$ that lies in the interior of $[0,T]$. If there were a zero $t_0 < \alpha$ before the singular region, choose the last such zero. Then $(t_0,\alpha)$ is a bang region ending at $\phi(\alpha) = \eta(\alpha) = 0$. The consequence above gives
\begin{align}
    \alpha-t_0 \ge \frac{\pi}{\Omega},
\end{align}
contradicting $0 < t_0 < \alpha < T \le \pi/\Omega$. Hence the part before $\alpha$ is either empty or a single bang region. 

Similarly, after $\beta$ there can be no later zero of $\phi$. Indeed, a bang region leaving $\beta$ starts with $\phi(\beta)=\eta(\beta)=0$, so in \eqref{sm_phi_bang_solution_AB} we have $B=0$. Hence
\begin{align}
    s\phi(\beta+\tau) = \frac{A}{\Omega^2} \bigl(1-\cos(2\Omega\tau)\bigr),
\end{align}
which cannot vanish for $0 < \tau < \pi/\Omega$ unless $A=0$, in which case the region would be singular rather than bang. Since $T-\beta<\pi/\Omega$, no later zero can occur. Thus the part after $\beta$ is either empty or a single bang region. Therefore, if a singular region is present, there is exactly one off region, with at most one bang region before it and at most one bang region after it. This completes the proof. 

\end{proof}

\subsection{Comparison of extremal candidates and proof of Theorem~\ref{thm:sm_det_opt}} \label{subsec:candidate-comparison}

As a final step, we now compare the extremal control families obtained in Proposition~\ref{prop:sm_extremal_structure}. In particular, we show that none of these candidates achieves a larger final $Y$-component than the symmetric two-bang control \eqref{sm_f_plus_det}.  This therefore completes the proof of Theorem~\ref{thm:sm_det_opt}.

\begin{proof}[Proof of Theorem~\ref{thm:sm_det_opt}]
By Lemmas~\ref{lem:sm_reduction_to_y} and \ref{lem:sm_sign_symmetry}, it is enough to maximize $y_f(T)$. By Proposition~\ref{prop:sm_extremal_structure}, we may restrict attention to two families of extremal controls: pure bang-bang controls with at most two switchings, and bang-off-bang controls. We compare these two families separately. In particular, we give an upper bound on $y(T)$ for these optimal candidates and show that the two-bang pulse saturates the bound.

First consider a pure bang-bang candidate, which has the form of
\begin{align}
U(T) = e^{-iH_s t_3}e^{-iH_{-s}t_2}e^{-iH_s t_1}, \qquad t_1+t_2+t_3=T .
\end{align}
Using $e^{-iH_s t} = \cos(\Omega t)I - i\frac{\sin(\Omega t)}{\Omega}(Z + s f_{\max}X)$, direct multiplication gives
\begin{align}
    y(T) = s\frac{2f_{\max}}{\Omega^2} \sin(\Omega t_2) \sin \left(\Omega(t_1-t_3)\right).
\end{align}
Since $\Omega t_2 + \Omega|t_1-t_3| \le \Omega T\le\pi$, 
\begin{align}
\label{sm_pure_bang_bound}
y(T) \le \frac{2f_{\max}}{\Omega^2}\sin(\Omega t_2)\sin (\Omega|t_1-t_3|) \le \frac{2f_{\max}}{\Omega^2} \sin^2\left(\frac{\Omega T}{2}\right).
\end{align}

Now consider a bang-off-bang candidate,
\begin{align}
U(T) = e^{-iH_{s_2}b}e^{-iZ\tau}e^{-iH_{s_1}a}, \qquad a+b+\tau=T, \qquad s_1,s_2\in\{+1,-1\}.
\end{align}
Direct multiplication gives
\begin{align}
\label{sm_bang_off_bang_y_formula}
y(T) = \frac{f_{\max}}{\Omega^2} \bigg[ \Omega\sin\tau \Big(s_1\sin(\Omega a)\cos(\Omega b) - s_2\sin(\Omega b)\cos(\Omega a)\Big) + (s_1-s_2)\cos\tau \sin(\Omega a)\sin(\Omega b) \bigg].
\end{align}
If $s_1=s_2=s$, then
\begin{align}
y(T) =  s\frac{f_{\max}}{\Omega} \sin\tau  \sin\left(\Omega(a-b)\right).
\end{align}
Since $\tau+\Omega|a-b| \le \tau+\Omega(a+b) = \Omega T-(\Omega-1)\tau \le \Omega T \le\pi$,  we get
\begin{align}
\label{sm_bob_same_sign_bound}
y(T) \le \frac{f_{\max}}{\Omega} \sin^2\left(\frac{\Omega T}{2}\right) \le \frac{2f_{\max}}{\Omega^2} \sin^2\left(\frac{\Omega T}{2}\right), 
\end{align}
where the last inequality uses $\Omega \le 2$.

It remains to consider $s_1=-s_2$. Write $s_1=s$ and $s_2=-s$, so that $y(T)=sY(a,b,\tau)$ where, with $u := a+b$, 
\begin{align}
Y(a,b,\tau) = \frac{f_{\max}}{\Omega}\sin\tau\sin(\Omega u) + \frac{2f_{\max}}{\Omega^2} \cos\tau\sin(\Omega a)\sin(\Omega b).
\end{align}
If $\cos \tau \ge 0$, then $Y \ge 0$, and the case $s=-1$ is automatically bounded above by zero. For $s=+1$, using $\sin(\Omega a)\sin(\Omega b) \le \sin^2(\Omega u/2)$ gives
\begin{align}
    y(T)\le G\left(u\right) := \frac{f_{\max}}{\Omega}\sin(T-u)\sin(\Omega u) + \frac{2f_{\max}}{\Omega^2}\cos(T-u) \sin^2\left(\frac{\Omega u}{2}\right).
\end{align}
A direct differentiation gives
\begin{align}
G'(u) = \frac{f_{\max}}{\Omega^2} \left(1+f_{\max}^2\cos(\Omega u)\right)\sin(T-u)\ge0,
\end{align}
where we used $f_{\max}\le1$, $0\le\Omega u\le\Omega T\le\pi$, and $0\le T-u\le T\le\pi$.
Thus $G$ is nondecreasing on $0\le u\le T$. Since the bang-off-bang candidate has $u=a+b\in[0,T]$, we obtain
Hence
\begin{align}
y(T) \le G(T) = \frac{2f_{\max}}{\Omega^2} \sin^2\left(\frac{\Omega T}{2}\right).
\end{align}
If $\cos\tau < 0$, it is also straightforward to get the upper bound of $\sin^2\left(\frac{\Omega T}{2}\right)$ for $s = \pm 1$. Therefore, every bang-off-bang candidate satisfies the same upper bound as in \eqref{sm_pure_bang_bound}. 

Finally, the symmetric two-bang control $f_+$ in \eqref{sm_f_plus_det} corresponds to the pure bang-bang choice with $s=+1, t_1=t_2=\frac{T}{2}, t_3=0$, in which we have
\begin{align}
y_{f_+}(T) = \frac{2f_{\max}}{\Omega^2} \sin^2\left(\frac{\Omega T}{2}\right),
\end{align}
saturating the upper bound. Therefore $f_+$ maximizes $y_f(T)$, and it solves Problem~\ref{prob:sm_det_oc}. This completes the proof of Theorem~\ref{thm:sm_det_opt}.
\end{proof}

\section{Randomized optimum for the single-qubit example}

In this section, we prove the randomized optimal-control statement used in the single-qubit example in the main text. We begin by stating the precise randomized optimization problem.

\begin{problem}[Randomized optimal control problem]
\label{prob:sm_rand_oc}
Consider the same controlled single-qubit Hamiltonian
\begin{align}
    H_f(t)=Z+f(t)X, \qquad |f(t)|\le f_{\max},
\end{align}
and the same target unitary channel
\begin{align}
    U_{\rm targ}=e^{-i\pi Y/2},\qquad \mathcal U_{\rm targ}(\rho) := U_{\rm targ}\rho U_{\rm targ}^\dagger.
\end{align}
A randomized control protocol consists of a finite ensemble of admissible controls $f^{(1)},\ldots,f^{(M)}$ and probabilities $\bm p = (p_1,\ldots,p_M)$, with $p_j\ge0$ and $\sum_{j=1}^M p_j=1$. The implemented channel is
\begin{align}
    \mathcal R_{\bm p, \{f^{(j)}\}}(\rho) := \sum_{j=1}^M p_j U_{f^{(j)}}(T)\rho U_{f^{(j)}}(T)^\dagger .
\end{align}
For fixed $T$ and $f_{\max}$, the randomized optimal-control problem is
\begin{align}
\label{sm_rand_oc_problem}
\epsilon_{\diamond}^{\rm rand,\star}(T) := \min_{\bm p, \{f^{(j)}\}}\frac12 \left\| \mathcal R_{\bm p, \{f^{(j)}\}} - \mathcal U_{\rm targ} \right\|_\diamond,
\end{align}
where the minimum is over all finite randomized ensembles of admissible controls.
\end{problem}

We now state the main theorem of this section.

\begin{theorem}[Randomized optimal control solution] \label{thm:sm_rand_opt}
Assume that $0<f_{\max} \le 1$ and $T\sqrt{1+f_{\max}^2} < \pi$. Let $f_+$ be the symmetric two-bang control from Theorem~\ref{thm:sm_det_opt}, and define
\begin{align}
    f_-(t):=-f_+(t).
\end{align}
Then the two-branch randomized protocol
\begin{align}
\label{sm_rand_two_branch_protocol}
f^{(1)}(t)=f_+(t),\qquad f^{(2)}(t)=f_-(t),\qquad p_1=p_2=\frac12
\end{align}
solves Problem~\ref{prob:sm_rand_oc}. Moreover, the corresponding minimum error is the square of the minimum deterministic error from Theorem~\ref{thm:sm_det_opt}:
\begin{align}
\label{sm_rand_equals_det_squared}
\epsilon_{\diamond}^{\rm rand,\star}(T) = \left(\epsilon_{\diamond}^{\rm det,\star}(T)\right)^2 = 1- \left( \frac{2f_{\max}}{1+f_{\max}^2} \sin^2\left(\frac{T\sqrt{1+f_{\max}^2}}{2} \right)\right)^2.
\end{align}
\end{theorem}

The proof has two ingredients. First, we derive a general lower bound on the diamond-distance error of any randomized control protocol using the Choi state of the relative channel. Second, we show that the two-branch protocol \eqref{sm_rand_two_branch_protocol} saturatse this lower bound.

We first begin with the general lower bound.

\begin{lemma}[Lower bound for mixed-unitary channels \cite{wallman2014randomized,braasch2026limits}] 
\label{lem:sm_rand_choi_lower_bound}
Let
\begin{align} 
\Phi(\rho) = \sum_j p_j W_j\rho W_j^\dagger
\end{align}
be a mixed unitary channel on a single qubit. Then
\begin{align}
\label{sm_rand_choi_lower_bound}
\frac12 \left\|\Phi- \mathcal I\right\|_\diamond \ge 1- \sum_j p_j \left| \frac12\Tr(W_j) \right|^2,
\end{align}
where $\mathcal I(\rho) := \rho$ is the identity channel. 
\end{lemma}

\begin{proof}
This lemma follows directly from Proposition 9 of \cite{wallman2014randomized}, together with the relation between the average gate fidelity and the entanglement fidelity \cite{nielsen2002a} and the Kraus representation of the latter. For completeness, however, we provide a direct proof.

Let $\ket{\Phi_2} := \frac{1}{\sqrt2} \left(\ket{00}+\ket{11}\right)$ be the Bell state, and define
\begin{align}
\rho_\Phi := (\Phi\otimes\mathcal I) \left(\ket{\Phi_2}\bra{\Phi_2}\right).
\end{align}
Using the definition of diamond norm, 
\begin{align}
\frac12 \left\|\Phi- \mathcal I \right\|_\diamond \ge \frac12 \left\|\rho_\Phi - \ket{\Phi_2}\bra{\Phi_2}\right\|_1.
\end{align}
For any state $\rho$ and pure state $\ket{\psi}$, the dual formulation of trace distance gives
\begin{align}
\frac12\left\|\rho-\ket{\psi}\bra{\psi}\right\|_1 = \max_{0\le M\le I} \Tr\left[M\left(\rho  - \ket{\psi}\bra{\psi}\right) \right] \ge 1-\bra{\psi}\rho\ket{\psi},
\end{align}
where we chose $M=I-\ket{\psi}\bra{\psi}$. 
Therefore
\begin{align}
\frac12 \left\|\Phi - \mathcal I\right\|_\diamond \ge 1- \bra{\Phi_2}\rho_\Phi\ket{\Phi_2}.
\end{align}
Finally,
\begin{align}
\bra{\Phi_2} \rho_\Phi\ket{\Phi_2} =  \sum_j p_j \left| \bra{\Phi_2} (W_j\otimes I) \ket{\Phi_2} \right|^2 = \sum_j p_j \left| \frac12\Tr(W_j) \right|^2.
\end{align}
This completes the proof. 
\end{proof}

We now apply Lemma~\ref{lem:sm_rand_choi_lower_bound} to the randomized optimal control problem in Problem~\ref{prob:sm_rand_oc}. This gives a universal lower bound on the diamond distance error in the randomized control setup.

\begin{proposition}[Universal randomized lower bound] \label{prop:sm_rand_lower_bound}

Assume that $0<f_{\max} \le 1$ and $T\sqrt{1+f_{\max}^2}\le\pi$. Define $\Omega:=\sqrt{1+f_{\max}^2}$ and 
\begin{align}
y^\star(T):= \frac{2f_{\max}}{1+f_{\max}^2} \sin^2\left(\frac{\Omega T}{2}\right).
\end{align}
Then every randomized protocol in Problem~\ref{prob:sm_rand_oc} satisfies
\begin{align}
\label{sm_rand_universal_lower_bound}
\frac12 \left\| \mathcal R_{\bm p,\bm f} - \mathcal U_{\rm targ} \right\|_\diamond \ge 1-y^\star(T)^2 . 
\end{align}
\end{proposition}

\begin{proof}
For an arbitrary randomized protocol $\mathcal R_{\bm p,\bm f}$, define the relative mixed-unitary channel
\begin{align}
\Phi := \mathcal U_{\rm targ}^\dagger\circ\mathcal R_{\bm p,\bm f}, \qquad \Phi(\rho) = \sum_j p_j W_j\rho W_j^\dagger, \qquad W_j:=U_{\rm targ}^\dagger U_{f^{(j)}}(T).
\end{align}
Since $U_{\rm targ} = -iY$, denoting $ U_{f^{(j)}}(T) = a_jI -i(x_jX+y_jY+z_jZ)$, then
\begin{align}
\frac12\Tr(W_j) = \frac12\Tr\left(iYU_{f^{(j)}}(T)\right) = y_j.
\end{align}
By the deterministic bound proved in Theorem~\ref{thm:sm_det_opt}, each branch
satisfies $|y_j|\le y^\star(T)$.

Applying Lemma~\ref{lem:sm_rand_choi_lower_bound} to $\Phi$, we obtain
\begin{align}
\frac12 \left\| \mathcal R_{\bm p,\bm f} - \mathcal U_{\rm targ} \right\|_\diamond = \frac12 \left\| \Phi-\mathcal I \right\|_\diamond  
&\ge  1- \sum_j p_j \left| \frac12\Tr(W_j) \right|^2 \\
&= 1-\sum_j p_j y_j^2 \\
&\ge 1-y^\star(T)^2.
\end{align}

\end{proof}

We are now ready to prove Theorem~\ref{thm:sm_rand_opt}.

\begin{proof}[Proof of Theorem~\ref{thm:sm_rand_opt}]
By Proposition~\ref{prop:sm_rand_lower_bound}, it remains only to show that the two-branch protocol in \eqref{sm_rand_two_branch_protocol} achieves the lower bound.

Let $\Omega:=\sqrt{1+f_{\max}^2}, \theta:=\frac{\Omega T}{2}, \eta:=y^\star(T) = \frac{2f_{\max}}{1+f_{\max}^2}\sin^2\theta$. For the two controls $f_+$ and $f_-=-f_+$, a direct calculation gives
\begin{align}
    U_{f_+}(T) = \chi I-i(\eta Y+\zeta Z), \qquad U_{f_-}(T) = \chi I-i(-\eta Y+\zeta Z),
\end{align}
where
\begin{align}
\zeta=\frac{\sin(2\theta)}{\Omega}, \qquad \chi=\frac{f_{\max}^2+\cos(2\theta)}{1+f_{\max}^2},
\end{align}
and
\begin{align}
    \chi^2+\eta^2+\zeta^2 = 1.
\end{align}
Let
\begin{align}
\mathcal R^\star(\rho) := \frac12 U_{f_+}(T) \rho U_{f_+}(T)^\dagger + \frac12 U_{f_-}(T) \rho U_{f_-}(T)^\dagger 
\end{align}
be the randomized channel implemented by the two-branch protocol in Theorem~\ref{thm:sm_rand_opt}, and define the corresponding relative channel $\widetilde{\mathcal R}^\star := \mathcal U_{\rm targ}^\dagger\circ \mathcal R^\star$.  Using $U_{\rm targ}=-iY$ and substituting the above expressions for $U_{f_+}(T)$ and $U_{f_-}(T)$, a direct calculation gives
\begin{align}
\widetilde{\mathcal R}^\star(\rho) = \eta^2\rho+(1-\eta^2)N\rho N, \qquad    N:=-\frac{\zeta X+\chi Y}{\sqrt{1-\eta^2}}.
\end{align}
Since $\chi^2+\zeta^2=1-\eta^2$, we have $N^\dagger=N, N^2=I, \Tr N=0$.  Thus, with $\mathcal N(\rho):=N\rho N$,
\begin{align}
\widetilde{\mathcal R}^\star-\mathcal I  = (1-\eta^2)(\mathcal N-\mathcal I).
\end{align}
Since $N$ is a unitary, we can apply Proposition~\ref{prop:sm_su2_diamond_distance} with $U = iN$ and $V = I$, and yield
\begin{align}
     \|\mathcal N-\mathcal I\|_\diamond = 2 \sqrt{1- \left| \frac12\Tr(iN)\right|^2} = 2.
\end{align}
Therefore, by unitary invariance of the diamond norm, 
\begin{align}
\frac12 \left\|\mathcal R^\star-\mathcal U_{\rm targ} \right\|_\diamond = \frac12 \left\| \widetilde{\mathcal R}^\star-\mathcal I \right\|_\diamond = 1-\eta^2 = 1-y^\star(T)^2.
\end{align}
Hence the two-branch protocol saturates the lower bound and solves Problem~\ref{prob:sm_rand_oc}. This completes the proof. 

\end{proof}

\section{Details on numerical simulation}

\subsection{Setup}
In this section, we provide details on the numerical simulations shown in Fig.~\ref{fig:randomized-cnot}. We used GRAPE to optimize a smooth surrogate cost based on the Frobenius norm, defined 
as $\|A\|_F:=\sqrt{\Tr(A^\dagger A)}$.  For a deterministic control waveform $\bm f$, the cost function is 
\begin{align}
\label{sm_G_det_cost}
    G_{\rm det}(\bm f) := \left\|U_{\bm f}(T)-U_{\rm targ}\right\|_F^2 .
\end{align}
For a randomized ensemble with branches $\bm f^{(1)},\ldots,\bm f^{(M)}$ and probabilities $p_1,\ldots,p_M$, we use the randomized surrogate
\begin{align}
\label{sm_G_rand_cost}
G_{\rm rand}(\bm p,\{\bm f^{(i)}\}) := \left\|\sum_{i=1}^M p_i U_{\bm f^{(i)}}(T) - U_{\rm targ} \right\|_F^2.
\end{align}  
We refer the reader to Ref.~\cite{khaneja2005optimal} for a review of the GRAPE algorithm.

The Frobenius surrogate is convenient for two reasons. First, it upper bounds the diamond-distance error through the averaged-unitary bound used in the main text. Indeed, by the sharpened mixing lemma~\cite{chen2021concentration},
\begin{align}
\epsilon_\diamond \le 2\left\|\sum_i p_i U_{\bm f^{(i)}}(T) - U_{\rm targ}\right\| \leq 2\left\|\sum_i p_i U_{\bm f^{(i)}}(T) - U_{\rm targ}\right\|_F.
\end{align}
Second, after discretizing the controls into piecewise-constant time steps, the final unitary is a smooth function of the control amplitudes. Hence $G_{\rm det}$ and $G_{\rm rand}$ are suitable gradient-based objectives for GRAPE. In our simulations, the diamond distance was computed only after optimization, as the final performance metric.

\subsection{Randomized GRAPE} \label{subsec:sm_rgrape}

\begin{algorithm}[t!]
\caption{Randomized GRAPE (R-GRAPE)}
\label{alg:rand_grape}

\KwIn{Hamiltonian model $H_0$, $\{H_k\}_{k=1}^m$, target unitary $U_{\rm targ}$, total time $T$, number of time steps $N$, number of branches $M$, amplitude bound $f_{\max}$, outer iterations $S$, inner GRAPE iterations $K$, and objective function $G$.}

\KwOut{Optimized branch controls $\{\bm f^{(i)}\}_{i=1}^M$ and probabilities $\bm p$.}

Initialize branch controls $\{\bm f^{(i)}\}_{i=1}^M$ satisfying the amplitude constraints
Initialize probabilities $\bm p$ with $p_i\ge0$ and $\sum_{i=1}^M p_i=1$

\For{$s\leftarrow 1$ \KwTo $S$}{

    \tcp{Forward propagate all branches}
    \For{$i\leftarrow 1$ \KwTo $M$}{
        $U^{(i)}(T)\leftarrow
        \FProp\left(H_0,\{H_k\},\bm f^{(i)},T,N\right)$
    }

    \tcp{Update probabilities with branch unitaries fixed}
    $\bm p\leftarrow
    \arg\min_{\bm p:p_i\ge0,\sum_i p_i=1}
    G\left(\bm p,\{U^{(i)}(T)\}_{i=1}^M\right)$

    \tcp{Update branch controls with probabilities fixed}
    \For{$\ell\leftarrow 1$ \KwTo $K$}{
        \For{$i\leftarrow 1$ \KwTo $M$}{
            $\bm f^{(i)}\leftarrow
            \FGR\left(H_0,\{H_k\},\bm f^{(i)},T,N;G,\bm p\right)$
            
            Project amplitudes of $\bm f^{(i)}$ back to
            $[-f_{\max},f_{\max}]$
        }
    }
}

\Return{$\{\bm f^{(i)}\}_{i=1}^M,\bm p$}
\end{algorithm}

We now describe the randomized GRAPE algorithm (R-GRAPE) used for the numerical comparison in Fig.~\ref{fig:randomized-cnot}. The algorithm is summarized in Algorithm~\ref{alg:rand_grape}. R-GRAPE extends standard GRAPE to a probabilistic mixture of $M$ branch unitaries $\{U^{(i)}(T)\}_{i=1}^M$, generated by controls $\{\bm f^{(i)}(\cdot)\}_{i=1}^M$. It alternates between updating the probability vector $\bm p$, with the branch controls fixed, and applying branch-wise GRAPE updates, with $\bm p$ fixed. Given its generality and effectiveness, it may be of independent interest beyond the specific examples considered here.

\subsection{Simulation parameters}
\label{subsec:sm_benchmarking_protocol}

For the CNOT simulations in Fig.~\ref{fig:randomized-cnot}, we used the two-qubit Hamiltonian in \eqref{two_qubit_model} with $J=1$ and $f_{\max}=10$. We compared deterministic GRAPE, the symmetry-generated randomized protocol constructed from the deterministic GRAPE pulse, and R-GRAPE. Deterministic GRAPE optimized the surrogate cost $G_{\rm det}$ in \eqref{sm_G_det_cost}, while R-GRAPE optimized the randomized surrogate cost $G_{\rm rand}$ in \eqref{sm_G_rand_cost}. For R-GRAPE, we used $M=4$ branches to match the number of branches in the symmetry-generated randomized protocol, $S=30$ outer probability updates, and $K=1000$ inner GRAPE steps per outer loop. For deterministic GRAPE, we used $30000$ GRAPE iterations, matching the total number of inner control-update steps used in R-GRAPE. All optimizer settings, including optimizer type, learning-rate schedule, time discretization, and amplitude projection rule, were kept fixed across methods. For each value of $T$, we performed 5 independent random initializations and report the best result. After optimization, the diamond distance was computed numerically by solving the standard SDP for the optimized solutions \cite{watrous2018the}.

\subsection{Error diagnostics}
To diagnose the cancellation mechanism, we decompose the relative error generator of the deterministic GRAPE pulse. For the optimized deterministic GRAPE pulse $\bm f$, define
\begin{align}
    U_{\rm targ}^\dagger U_{\bm f}(T) = e^{-iE^\star},
\end{align}
and expand
\begin{align}
E^\star = \sum_{P\in\mathcal P_2} c_P P, \qquad c_P=\frac{1}{4}\Tr(PE^\star), \qquad \mathcal P_2:=\{I,X,Y,Z\}^{\otimes 2}.
\end{align}
Fig.~\ref{fig:sm_cnot_error_decomp} shows that the dominant component is $Z_1X_2$, labeled as $ZX$ in the figure. This component is therefore canceled by the $X_2$ time-reversal pairing explained in the main text. This explains the quadratic suppression observed in Fig.~\ref{fig:randomized-cnot}.

\begin{figure}[t!]
    \centering
    \includegraphics[width=0.75\linewidth]{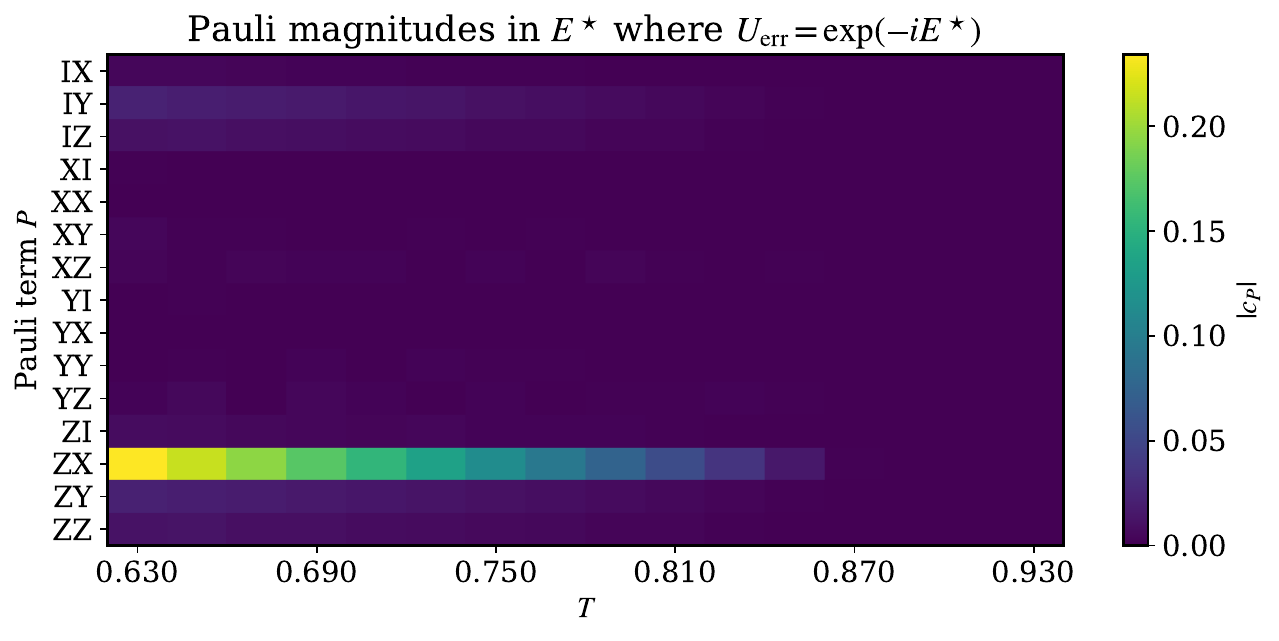}
    \caption{Pauli decomposition of the relative error generator $E^\star$ for the optimized deterministic GRAPE pulse in Fig.~\ref{fig:randomized-cnot}. The identity component is omitted. The dominant component is $Z_1X_2$, which is cancelled by the time-reversal pairing method discussed in the main text. This explains the quadratic improvement. }
    \label{fig:sm_cnot_error_decomp}
\end{figure}

\subsection{Additional numerical simulation on estimating observables in an experiment}

In this subsection, we numerically test the CNOT protocols used in Fig.~\ref{fig:randomized-cnot} in a setup where one estimates expectation values of observables. We first note that the diamond-distance error also controls observable-estimation errors. Indeed, for any two channels $\mathcal E$ and $\mathcal F$, any input state $\rho$, and any observable $O$ with $\|O\|\le 1$,
\begin{align}
\left| \Tr\left(O\left(\mathcal E(\rho)-\mathcal F(\rho)\right)\right)\right|
\le\left\|\mathcal E(\rho)-\mathcal F(\rho)\right\|_1
\le\left\|\mathcal E-\mathcal F\right\|_\diamond.
\end{align}
The diamond distance therefore gives a worst-case guarantee for observable estimation (while the following numerical test probes the error for specific states and observables).

\begin{figure}
    \centering
    \includegraphics[width=0.9\linewidth]{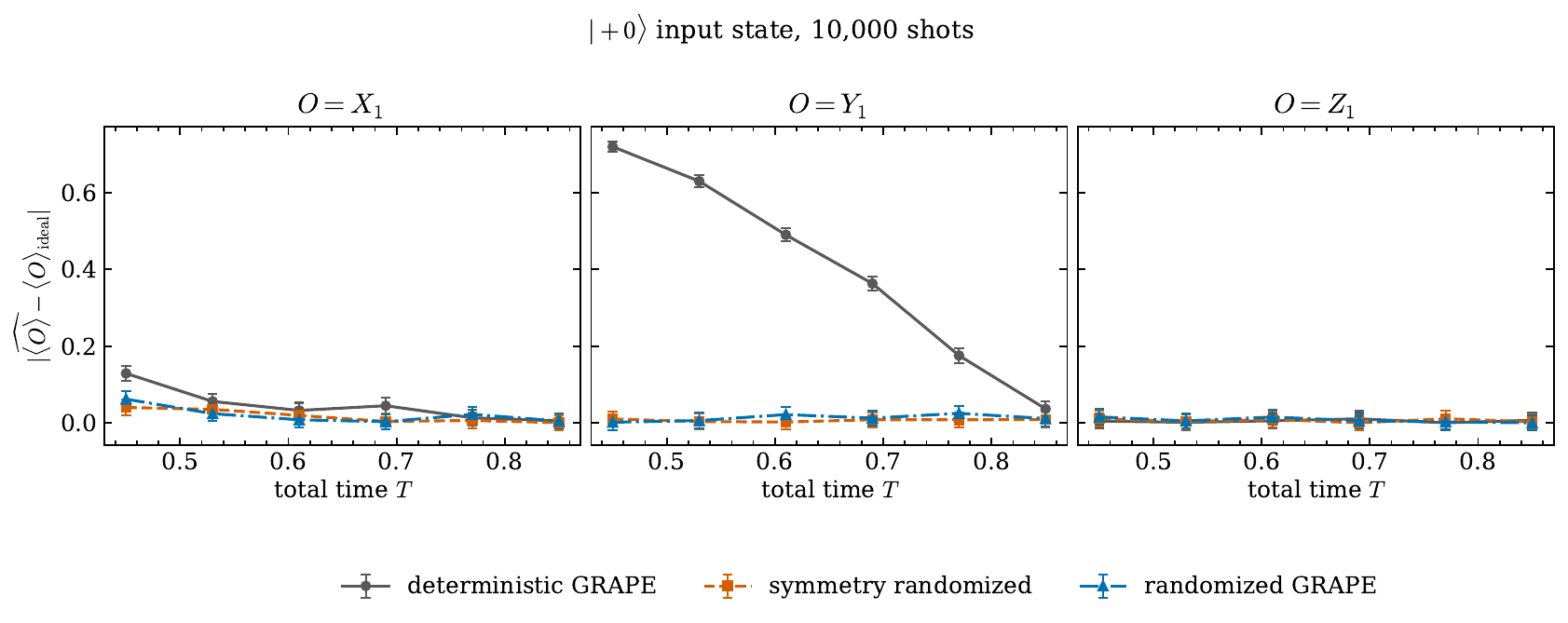}
    \caption{Performance comparison between three different CNOT control protocols used in Fig.~\ref{fig:randomized-cnot} for estimating expectation values of Pauli observables. The input state is $\rho_{\rm in}=\ket{+0}\bra{+0}$, for which the ideal CNOT output is the Bell state, and each point uses $N_{\rm shots}=10^4$ shots. The deterministic GRAPE pulse shows a large bias in the $Y_1$ observable, while the randomized protocols suppress this bias. The $X_1$ obserable shows a weaker improvement, whereas $Z_1$ is comparatively insensitive to the dominant coherent error. This behavior is consistent with the leading deterministic error being in the $Z_1X_2$ error: $Y_1$ is first-order sensitive to this error, while $X_1$ and $Z_1$ are not.}
    \label{fig:sm_cnot_observable}
\end{figure}

For each total time $T$, we prepare
\begin{align}
    \rho_{\rm in}=\ket{+0}\bra{+0}
\end{align}
and apply the approximate CNOT controls to prepare a Bell state, and then estimate the expectation value of a given observable. We compare deterministic GRAPE, the symmetry-generated randomized protocol, and R-GRAPE.

The estimation procedure is as follows.
\begin{enumerate}
    \item[1.] Choose a protocol. For deterministic GRAPE, every shot uses the same optimized control $\bm f$. For the symmetry-generated randomized protocol, every shot samples uniformly from the four branches in \eqref{concatenated_control_protocol_CNOT}. For R-GRAPE, every shot samples from the optimized ensemble according to the optimized probability vector $\bm p$.

    \item[2.] For each shot $s=1,\ldots,N_{\rm shots}$, prepare $\rho_{\rm in}$, apply the sampled control branch, and measure $O$. For simplicity, assuming $O$ is Pauli, the outcome is $a_s\in\{\pm1\}$.

    \item[3.] Estimate the expectation value of $O$ by
    \begin{align}
       \widehat {\langle O\rangle} =  \frac{1}{N_{\rm shots}} \sum_{s=1}^{N_{\rm shots}} a_s.
    \end{align}
    We then report the observable-estimation error
     \begin{align}
        \left| \widehat{\langle O\rangle} - \langle O\rangle_{\rm ideal}
        \right|, \qquad
        \langle O\rangle_{\rm ideal}=\Tr\left(O\text{CNOT}\rho_{\rm in}\text{CNOT}^{\dagger}\right).
    \end{align}
\end{enumerate}

Fig.~\ref{fig:sm_cnot_observable} shows the resulting observable-estimation errors for $O\in\{X_1,Y_1,Z_1\}$ and $N_{\rm shot} = 10000$. The most advantage appears in the $Y_1$ observable: the deterministic GRAPE pulse has a visible systematic bias, whereas both randomized protocols suppress this bias to the finite-shot level. The $X_1$ observable shows a smaller but still visible improvement, while the $Z_1$ observable is nearly insensitive to the error component that dominates the deterministic pulse.

This behavior can be understood from the first-order response to the coherent relative error. As in the previous subsection, define the error generator
\begin{align}
    U_{\rm targ}^\dagger U_{\bm f}(T)=e^{-iE^\star},
    \qquad E^\star=\sum_{P\in\mathcal P_2}c_P P, \qquad \mathcal P_2=\{I,X,Y,Z\}^{\otimes 2}.
\end{align}
For an observable $O$, define 
\begin{align}
    A_O:=U_{\rm targ}^\dagger O U_{\rm targ}.
\end{align}
Then, for the deterministic pulse,
\begin{align}
\label{eq:observable_linear_response}
    \Tr\left(OU_{\bm f}(T)\rho_{\rm in}U_{\bm f}(T)^\dagger\right) = \Tr[A_O\rho_{\rm in}] + i\Tr\left([E^\star,A_O]\rho_{\rm in}\right) + \mathcal O(\|E^\star\|^2).
\end{align}
For $\rho_{\rm in}=\ket{+0}\bra{+0}$, only Pauli strings in $\mathrm{span}\{I,X_1,Z_2,X_1Z_2\}$ have nonzero expectation, and for these operators, we have
\begin{align}
    A_{X_1}=X_1X_2,\qquad
    A_{Y_1}=Y_1X_2,\qquad
    A_{Z_1}=Z_1.
\end{align}
The error diagnostic in Fig.~\ref{fig:sm_cnot_error_decomp} shows that the dominant coherent error of the deterministic GRAPE pulse is in the $Z_1X_2$ error. Therefore the first-order observable response in \eqref{eq:observable_linear_response} is determined by commutators with $Z_1X_2$.

For $Y_1$, this first-order response is nonzero:
\begin{align}
    [Z_1X_2,A_{Y_1}] = [Z_1X_2,Y_1X_2] \propto X_1,
\end{align}
and $\Tr(X_1\rho_{\rm in})\neq 0$.  Thus $Y_1$ is directly sensitive to the dominant deterministic coherent error.

By contrast, the same dominant error gives no first-order contribution to $X_1$ or $Z_1$.  Indeed,
\begin{align}
    [Z_1X_2,A_{X_1}] =  [Z_1X_2,X_1X_2] \propto Y_1, \qquad  \Tr(Y_1\rho_{\rm in})=0,
\end{align}
while
\begin{align}
    [Z_1X_2,A_{Z_1}] = [Z_1X_2,Z_1] = 0.
\end{align}
Hence both $X_1$ and $Z_1$ are insensitive to the dominant error at first order. Any improvement in $X_1$ therefore comes from smaller subleading error components, whereas the large improvement in $Y_1$ reflects cancellation of the dominant $Z_1X_2$ component by the randomized protocols.  This explains the trend in Fig.~\ref{fig:sm_cnot_observable}.

\section{Finite-width randomized Pauli boundary pulses}

In the main text, the randomized boundary pulse construction was first described for instantaneous boundary pulses. Here we show that, for Pauli boundary operations, the leading finite-pulse contribution can also be canceled by randomization for 1-local coherent noise.

Consider a pulse sequence $w$ 
whose ideal endpoint implements a Pauli operation $P$ up to a global phase. In the presence of a weak coherent noise $\eta V$, we denote the perturbed final unitary as
\begin{align}
U_w^\eta = P\exp\left(-i\eta F_w(V) + \mathcal O(\eta^2)\right),
\end{align}
where $F_w(V)$ is the first-order error action accumulated in the toggling frame of the pulse sequence. A randomized finite-width implementation cancels the leading boundary error if $\mathbb E_w[F_w(V)]=0$.

We start with a primitive $\pi$ pulse about a single-qubit Pauli $A\in\{X,Y,Z\}$,
\begin{align}
A_s:\qquad U_{A_s}(t)=e^{-is\theta(t)A}, \qquad s=\pm1,
\end{align}
where $\theta(0)=0$ and $\theta(\tau_p)=\pi/2$. The two signs implement the same Pauli operation $A$ up to a global phase, but traverse opposite paths during the finite pulse. We assume
\begin{align}
C := \int_0^{\tau_p}\cos(2\theta(t))dt = 0,
\end{align}
which holds for any pulse with a time-symmetric envelope $\dot\theta(t)=\dot\theta(\tau_p-t)$, including square and Gaussian pulses. Define also
\begin{align}
S := \int_0^{\tau_p}\sin(2\theta(t))dt.
\end{align}

\begin{lemma}[Single-qubit randomized Pauli pulse sequences]
\label{lem:finite_width_single_qubit_gadget}
Let $P\in\{X,Y,Z\}$ be a desired single-qubit Pauli operation. Choose two anticommuting Paulis $A$ and $B$ such that $P\sim BA$, where $\sim$ denotes equality up to a global phase, and define
\begin{align}
\mathcal W_{P;A,B} = \left\{A_+B_-, B_+A_-, A_-B_+, B_-A_+ \right\}.
\end{align}
Every sequence in $\mathcal W_{P;A,B}$ implements $P$ up to a global phase. Moreover, for any traceless single-qubit noise $V$,
\begin{align}
\frac14\sum_{w\in\mathcal W_{P;A,B}}F_w(V)=0 .
\end{align}
\end{lemma}

\begin{proof}
It is enough to prove the claim for a Pauli error $R\in\{X,Y,Z\}$ and then use linearity. For a primitive pulse, a direct calculation gives
\begin{align}
F_{A_s}(R) = \int_0^{\tau_p}U_{A_s}(t)^\dagger R U_{A_s}(t)dt =
\begin{cases}
\tau_p R, & [A,R]=0,\\
sS iAR, & \{A,R\}=0.
\end{cases}
\end{align}
For a two-pulse sequence, with the left pulse applied first, the first-order error is
\begin{align}
F_{A_sB_t}(R) = F_{A_s}(R) + A F_{B_t}(R) A,
\end{align}
and similarly for the reversed order. Averaging the four choices in $\mathcal W_{P;A,B}$, after elementary calculations, gives
\begin{align}
\frac14\sum_{w\in\mathcal W_{P;A,B}}F_w(R) =
\begin{cases}
2 \tau_p R, & [A,R]=[B,R]=0,\\
0, & \text{otherwise.}
\end{cases}
\end{align}
On a single qubit, no non-identity Pauli commutes with two anticommuting non-identity Paulis $A$ and $B$. Hence the average vanishes for every $R\in\{X,Y,Z\}$, and therefore for every traceless $V$.

\end{proof}

Lemma~\ref{lem:finite_width_single_qubit_gadget} gives the local building block, 
which can be applied independently on each qubit to implement multi-qubit Pauli strings under 1-local noise, except for identity factors. For instance, to implement $X_1=X\otimes I$, applying the randomized $X$ pulse sequences on qubit 1 while leaving qubit 2 idle would allow 1-local noise on qubit 2 to accumulate a first-order boundary error. A simple remedy is to replace each idle identity by an active randomized identity, e.g. $I_{\rm rand}:=X_{\rm rand}X_{\rm rand}$. To match the pulse time for for every qubit, we therefore use the uniform local replacement
\begin{align}
I &\longmapsto X_{\rm rand}X_{\rm rand},&
X &\longmapsto Y_{\rm rand}Z_{\rm rand},\\
Y &\longmapsto Z_{\rm rand}X_{\rm rand},&
Z &\longmapsto X_{\rm rand}Y_{\rm rand},
\end{align}
where each $P_{\rm rand}$ denotes the randomized Pauli pulse sequences of Lemma~\ref{lem:finite_width_single_qubit_gadget}. The final unitaries are correct up to phase, since $XX=I$, $ZY\sim X$, $XZ\sim Y$, and $YX\sim Z$. Then, these randomized pulse sequences are robust against 1-local noise on average. 

\begin{proposition}[Finite-width Pauli-string pulse sequences for 1-local noise]
\label{prop:finite_width_pauli_string_gadget}
Let $G=\bigotimes_{j=1}^n G_j, ~G_j\in\{I,X,Y,Z\}$ be an $n$-qubit Pauli boundary pulse operation, and let the noise be 1-local,
\begin{align}
V=\sum_{j=1}^n\sum_{P\in\{X,Y,Z\}}v_{j,P}P_j.
\end{align}
Using the local replacement rule above independently on each qubit gives a randomized finite-width implementation of $G$ such that every branch implements $G$ up to a global phase and
\begin{align}
\mathbb E[F_G(V)]=0 .
\end{align}
Consequently, the averaged boundary contribution exhibits an averaged error scaling of $\mathcal O(\eta^2\tau_p^2)$.
\end{proposition}

\begin{proof}
It suffices to consider a one-local Pauli error $R_j\in\{X_j,Y_j,Z_j\}$. In the local replacement rule, the implementation on qubit $j$ consists of two independently randomized Pauli pulse sequences, $P_{j,1}^{\rm rand}$ and $P_{j,2}^{\rm rand}$. For fixed branches $\beta_1,\beta_2$, the first-order error has the same form
\begin{align}
F_{j,\beta_1,\beta_2}(R_j) = F_{\beta_1}^{P_{j,1}}(R_j) + P_{j,1}F_{\beta_2}^{P_{j,2}}(R_j)P_{j,1}.
\end{align}
By Lemma~\ref{lem:finite_width_single_qubit_gadget}, each term has zero average over its own branch, and hence
\begin{align}
\mathbb E_{\beta_1,\beta_2}F_{j,\beta_1,\beta_2}(R_j)=0 .
\end{align}
Pulses on other qubits commute with $R_j$, so this local cancellation is unchanged in the full tensor-product implementation. Therefore $\mathbb E[F_G(P_j)]=0$ for all $j,\mu$, and the result follows by linearity in $V$.
\end{proof}

\begin{figure}[t!]
    \centering
    \includegraphics[width=0.8\linewidth]{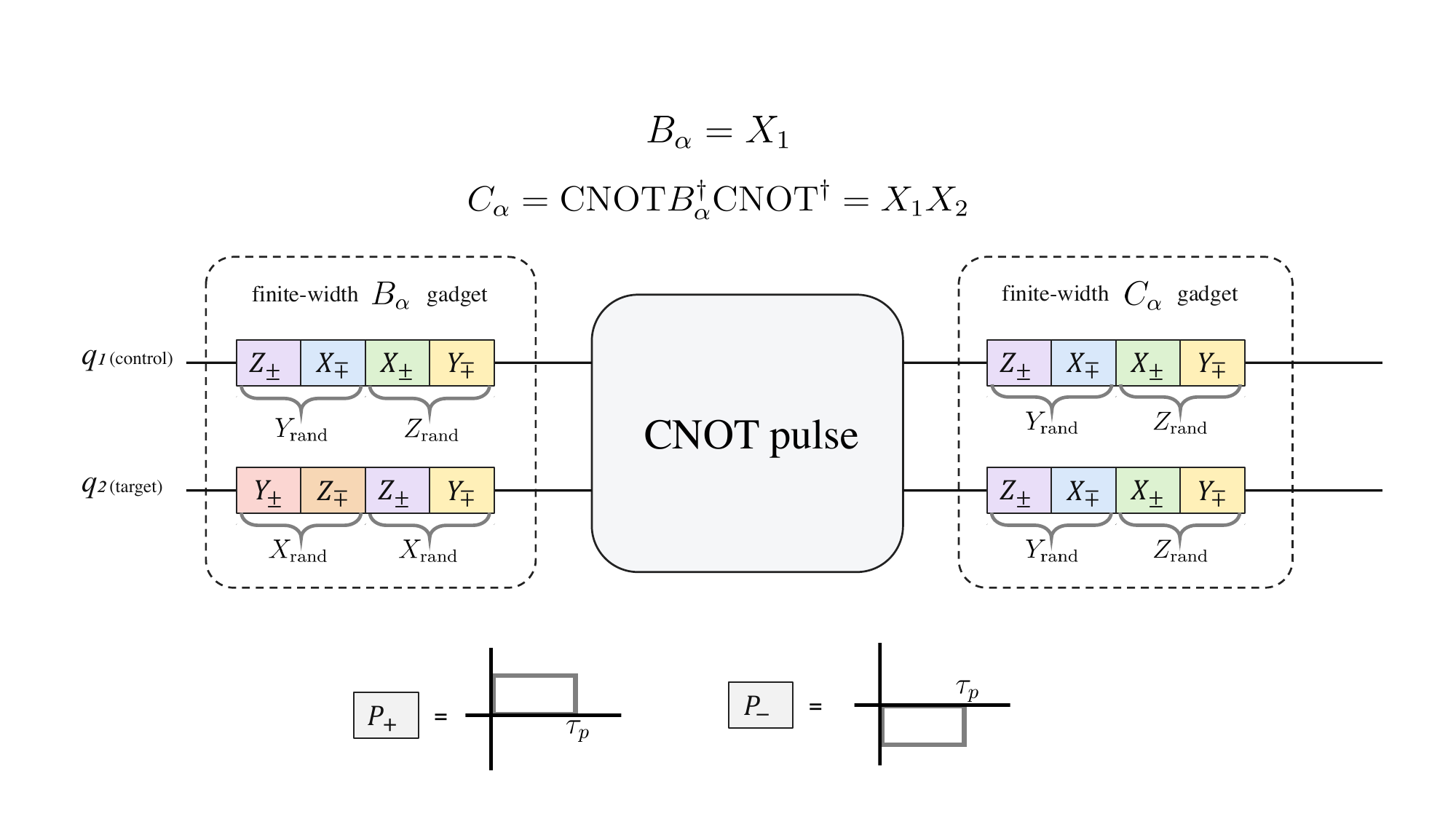}
    \caption{Schematic of a single branch of the finite-width randomized boundary-pulse construction for CNOT. The reference CNOT pulse is preceded by a Pauli-string boundary pulse sequence 
    $B_\alpha$ and followed by the conjugated pulse sequence $C_\alpha=U_{\rm CNOT}B_\alpha^\dagger U_{\rm CNOT}^\dagger$. Twirling over all two-qubit Pauli choices of $B_\alpha$, together with the corresponding $C_\alpha$, cancels the coherent error to first order. Here we show one representative branch with $B_\alpha=X_1$. Each primitive pulse $P_s$ denotes evolution under $H(t)=s\dot{\theta}(t)P$ for duration $\tau_p$, with $s=\pm 1$; for the schematic, we depict these primitives as square pulses. The two-pulse sequence $P_{\rm rand}$ are sampled from the four sequences $\{A_+B_-,B_+A_-,A_-B_+,B_-A_+\}$, where $A$ and $B$ anticommute and $P\sim BA$.}
    \label{fig:cnot_example_boundary}                                                      
\end{figure}

Using Proposition~\ref{prop:finite_width_pauli_string_gadget}, Clifford targets admit a general finite-width boundary construction that removes all first-order contributions in $\eta$.

\begin{corollary}[Application to Clifford targets] \label{cor:finite_width_clifford_boundary}
Let $U_{\rm targ}$ be a Clifford unitary and choose the boundary unitaries $B_\alpha$ to be Pauli strings. Then $C_\alpha=U_{\rm targ}B_\alpha^\dagger U_{\rm targ}^\dagger$ is also a Pauli string up to a phase. Therefore both $B_\alpha$ and $C_\alpha$ can be implemented using the randomized finite-width Pauli-string pulse sequences of Proposition~\ref{prop:finite_width_pauli_string_gadget}. For coherent 1-local noise, this removes the $\mathcal O(\eta\tau_p)$ finite-boundary contribution in the randomized boundary protocol. 

Moreover, by choosing $B_\alpha$ uniformly from the $n$-qubit Pauli group, the reference-control
first-order noise generator is averaged out as $\frac{1}{4^n}\sum_{B\in\{I,X,Y,Z\}^{\otimes n}}B^\dagger F_VB = 0$. Hence the reference control $\mathcal O(\eta)$ term and the finite-boundary $\mathcal O(\eta\tau_p)$ term are both canceled, so the entire randomized boundary protocol has diamond distance error scaling $\mathcal O(\eta^2)$.
\end{corollary}

For generating Fig.~\ref{fig:boundary_noise_robustness_CNOT} in the main text, we used Corollary~\ref{cor:finite_width_clifford_boundary}. The reference CNOT pulse was obtained using deterministic GRAPE, with the total time $T$ chosen sufficiently large so that the target gate is implemented to numerical precision in the absence of noise. The 16 boundary pulses $B_\alpha$ and $C_\alpha = U_{\rm targ} B_\alpha^\dagger U_{\rm targ}^\dagger$ were then constructed from the finite-width Pauli pulse sequences described above using square pulses. Fig.~\ref{fig:cnot_example_boundary} illustrates the general idea of this construction for a representative branch with $B_\alpha=X_1$ and the corresponding $C_\alpha=X_1X_2$, showing how the boundary unitaries are implemented using square-pulse Pauli gadgets.

\end{document}